\documentclass[a4paper,USenglish,cleveref, autoref, thm-restate, pdfa]{lipics-v2021}
\usepackage[utf8]{inputenc}

\usepackage{makecell,float,graphicx,xcolor,listings}
\usepackage{hyperref}
\hypersetup{
	colorlinks,
	linkcolor={green!70!black},%
	citecolor={blue!70!black},%
	urlcolor={blue!80!black}
}

\newcommand{\eps}{\varepsilon}
\newcommand{\disc}{ball}
\newcommand{\adisc}{\alpha\text{-}ball}

\newcommand{\R}{\mathbb{R}}

\DeclareMathOperator{\polylog}{polylog}
\DeclareMathOperator{\vol}{Vol}

\pdfoutput=1 %
\hideLIPIcs
\nolinenumbers

\def\WITHAPPENDIX{}     %

\title{Approximating Multiplicatively Weighted Voronoi Diagrams: Efficient Construction with Linear Size}
\titlerunning{Approximate Multiplicatively Weighted Voronoi Diagrams with Optimal Size} 

\author{Joachim Gudmundsson}{University of Sydney, Australia}{joachim.gudmundsson@gmail.com}{https://orcid.org/0000-0002-6778-7990}{}
\author{Martin P. Seybold}{University of Vienna, Austria}{mpseybold@gmail.com}{https://orcid.org/https://orcid.org/0000-0001-6901-3035}{}
\author{Sampson Wong}{University of Copenhagen, Copenhagen}{sampson.wong123@gmail.com}{https://orcid.org/0000-0003-3803-3804}{}

\authorrunning{Joachim Gudmundsson, Martin P. Seybold and Sampson Wong} 
\Copyright{Joachim Gudmundsson, Martin P. Seybold and Sampson Wong} %


\category{} %
\relatedversion{See \url{https://arxiv.org/abs/2112.12350} for the full version of the paper.}

\acknowledgements{
This work was supported under the Australian Research Council Discovery Projects funding scheme (project number DP180102870)
and partially supported by Starting Grant 1054-00032B from the Independent Research Fund Denmark under the Sapere Aude research career programme.
This project has received funding from the European Research Council (ERC)  under the European Union's Horizon 2020
\marginpar{\includegraphics[height=22pt]{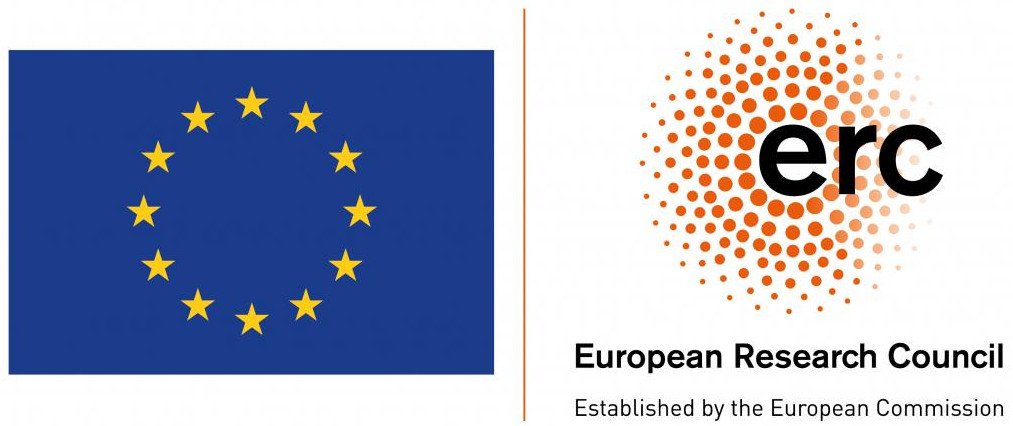}}
research and innovation programme (Grant agreement No. 101019564) and the Austrian Science Fund (FWF) project Z~422-N, project I~5982-N, and project P~33775-N, with additional funding from the \textit{netidee SCIENCE Stiftung}, 2020--2024.
}

\EventEditors{Wolfgang Mulzer and Jeff M. Phillips}
\EventNoEds{2}
\EventLongTitle{40th International Symposium on Computational Geometry (SoCG 2024)}
\EventShortTitle{SoCG 2024}
\EventAcronym{SoCG}
\EventYear{2024}
\EventDate{June 11-14, 2024}
\EventLocation{Athens, Greece}
\EventLogo{socg-logo.pdf}
\SeriesVolume{293}
\ArticleNo{XX}     %

\begin{document}
\maketitle

\begin{abstract}
Given a set of $n$ sites from $\mathbb{R}^d$, each having some positive weight factor, the Multiplicatively Weighted Voronoi Diagram is a subdivision of space that associates each cell to the site whose weighted Euclidean distance is minimal for all points in the cell.

We give novel approximation algorithms that output a cube-based subdivision such that the weighted distance of a point with respect to the associated site is at most $(1+\varepsilon)$ times the minimum weighted distance, for any fixed parameter $\varepsilon \in (0,1)$.
The diagram size is $O_d(n \log(1/\varepsilon)/\varepsilon^{d-1})$ and the construction time is within an $O_D(\log(n)/\varepsilon^{(d+5)/2})$-factor of the size bound.
We also prove a matching lower bound for the size, showing that the proposed method is the first to achieve \emph{optimal size}, up to $\Theta(1)^d$-factors.
In particular, the obscure $\log(1/\varepsilon)$ factor is unavoidable.
As a by-product, we obtain 
a factor $d^{O(d)}$ improvement in size for the unweighted case
and 
$O(d \log(n) + d^2 \log(1/\varepsilon))$  point-location time in the subdivision, improving the known query bound by one $d$-factor.

The key ingredients of our approximation algorithms are the study of convex regions that we call cores, an adaptive refinement algorithm to obtain optimal size, and a novel notion of \emph{bisector coresets}, which may be of independent interest.
In particular, we show that coresets with $O_d(1/\varepsilon^{(d+3)/2})$ worst-case size can be computed in near-linear time.
\end{abstract}

\section{Introduction}      %
Voronoi Diagrams are structures of fundamental importance for many scientific fields.
In particular, planar variants with linear worst-case size are very well understood (e.g.~\cite{voronoi-book,dutch-book}).
Though closely related to the Nearest-Neigbhor search problem, the \emph{explicit} subdivisions provided by Voronoi Diagrams are a central tool for various problems, including meshing in scientific computing, planning of facility locations, motion planning, or surface reconstruction. %

Given a set of sites $\{s_1,\ldots,s_n\} \subset \R^d$, each having a positive weight $w_i > 0$, 
their Multiplicatively Weighted Voronoi Diagram~(MWVD) is the subdivision of $\R^d$ into cells that associates each cell to one site, i.e. the site~$s_i$ that minimizes $\lVert p - s_i\rVert_2/w_i$ for all points $p$ in the cell. 
Though all bisectors in an MWVD are either half-spaces ($w_i=w_j$) or Apollonian spheres ($w_i \neq w_j$), the two main difficulties with MWVDs are that Voronoi regions may contain holes, %
and that the multiplicative weights can violate the triangle inequality. %

The MWVD in $\R^1$ has linear size and can be obtained using a Divide~\&~Conquer algorithm in $O(n \log n)$ time~\cite{AURENHAMMER1986}.
Aurenhammer and Edelsbrunner showed that MWVDs in $\mathbb R^2$ can have $\Omega(n^2)$ size and gave a worst-case optimal algorithm~\cite{AurenhammerE84}.
Held and de Lorenzo~\cite{HeldL20} gave a sweep approach for $2$D that runs in $O(n^2 \log n)$ time.
In special cases, $2$D MWVD size is known to have near-linear, or even linear, bounds~\cite{Har-PeledR15,FanR20}.
In general, unweighted Voronoi Diagrams, i.e. all $w_i=1$, are well known to have $\Omega(n^{\lceil d/2 \rceil})$ worst-case size (see e.g.~\cite{sariel-book}).

\subparagraph*{Importance of cube-based Approximate Voronoi Diagrams.}
We limit our discussion on two applications where the simplicity of cube-based AVDs is key for strong bounds. 

\begin{enumerate}[(i)]
    \item {\bf Axis-Aligned Segment-Queries in $2$D.} 

Using Chazelle's Point-Location \& Walk method~\cite[Sect. 4.2]{Chazelle86} on an $2$D MWVD, it is possible to traverse all $k$ cells of the subdivision that are intersected by an axis-aligned query line-segment in $O(\log(n) + k)$ time,
which determines the $\Omega(k)$ distinct nearest-sites for (the sequence of points that are contained in) the query-segment.

Now, an approximate subdivision that consists of canonical squares, or set difference of canonical squares, allows to merge common boundaries of adjacent squares, associated to the same Voronoi site, without increasing the size bound of the subdivision.
Thus, allowing to retain the $O(\log(n) +k)$ query bound in the approximate setting.

\item {\bf Fast Point-Queries when $d$ is large.}
The `curse of dimensionality' typically refers to the broad phenomena that either the query-bounds or the space-bounds of known structures for (exact) nearest-neighbor search deteriorate `quickly' as $d$ increases.
In $\eps$-approximate nearest-neighbor search, we are mainly interested in the range $d=2$ to $d=O(\log(n)/\eps^2)$, due to Johnson-Lindenstrauss dimension reduction (see, e.g.,~\cite{sariel-book, Har-PeledIM12}).

Now, cube-based subdivisions allow to use compressed QuadTrees to obtain \emph{very strong}
query bounds. %
For example, in a subdivision of $\R^d$ with $N=O(n/\eps^d)$ cubes, the query time is $O(d \log(n/\eps^d))=O(d \log(n) +d^2\log(1/\eps))$.\\
In contrast, query bounds containing $O(1)^d$-terms are only fast when $d$ is \emph{very small}.
\end{enumerate}

For careful comparison with respect to the dimension, we distinguish between $O$-notation, 
$O_D$-notation that assumes a `constant-dimension' and hides $d^{O(d)}$-factors, 
and $O_d$-notation that assumes a `small-dimension' and hides $O(1)^d$-factors.
E.g. $O((8d)^d)=O_d(d^d) = O_D(1)$.
Note that there is a separation between space bounds in the $O_D$-regime and the $O_d$-regime.
For $d=O(\log \log n)$, any $O(1)^d$ factors in size are $O(\polylog n)$ factors, whereas $d^d$-factors are $\omega(\polylog n)$.
Further, $c^d$-factors in size are sub-linear $O(n^{1/p})$ for $d \leq \log_c(n)/p$, unlike $d^d$-factors.

This work studies the problem of computing $\eps$-Approximate MWVDs for prescribed $\eps>0$.
That is, a subdivision of $\R^d$ into cells that are cubes, or set-difference of cubes, that associates each cell with one site that is an $\eps$-approximate weighted nearest-neighbor for all points in the cell.
The only known solution til date is to employ the, more general, framework of Har-Peled and Kumar~\cite{Har-PeledK15},
which, e.g., found application in the work~\cite{AronovBK20}.

\subparagraph*{Contribution and Paper Organization.}

Our approach considers convex regions that we call `cores', which are the intersection of at most $n-1$ Apollonian balls of MWVD bisectors.
In Section~\ref{sec:core_algorithm}, we introduce an \textsf{Adaptive Refinement} algorithm that $\eps$-approximates each core with a set of $d$-cubes, and show that each core is $\eps$-approximated with $O_d(\log(1/\eps)/\eps^{d-1})$ cubes.
In Section~\ref{subsec:amwvd_from_approximate_cores}, we show that a top-down propagation in the compressed QuadTree over the set of $d$-cubes allows to obtain an $\eps$-AMWVD that consists of $O_d( n \log(1/\varepsilon)/\varepsilon^{d-1})$
cells that are $d$-cubes, or the set difference of $d$-cubes, each of which associated to \emph{one site} that is weighted nearest-neighbor for all points in the cell, up to a $(1+\eps)$ factor.
One by-product of our construction is thus a compressed QuadTree that can report an $\eps$-NN of a query-point in $O(d \log (n) + d^2\log(1/\varepsilon))$~time, thus improving on the query-time of the structure from~\cite{Har-PeledK15} by one $d$-factor.

We prove a matching lower bound on the size of the subdivision in Section~\ref{sec:lowerbound}.
Specifically, we show that \emph{every} subdivision of $\R^d$, formed by axis-aligned hyper-rectangles, that is an $\eps$-approximation of an Apollonian ball must contain $\Omega_d( \log(1/\eps)/\eps^{d-1})$ hyper-rectangles.
Our proposed bound improves on the known $\Omega_d(\eps/(\eps\sqrt{d})^d)$ bound from~\cite{AryaM02,AryaMM09} in two ways.
First, the denominator is free of the $\sqrt{d}$-factor and, second, it is the first known lower bound that shows that a $\log(1/\eps)$-factor is \emph{required} in the space.
Thus, the proposed construction is the first that computes an $\eps$-AMWVD with worst-case optimal size, up to $\Theta_d(1)$-factors.

In Section~\ref{sec:sspd_section}, we introduce our second approximation algorithm which is the key component to improve the construction time from quadratic to near-linear.
We show that cores admit an $\eps$-approximation with low complexity, i.e. with $O_d(1/\eps^{(d+3)/2})$ bisectors, and give an algorithm that outputs such bisector coresets in 
$O_D(n \log (n)/\eps^{3(d+1)/2})$ time, based on an $O(1/\eps)$-Semi-Separated Pair Decomposition~(SSPD) of the site locations.
If the sites are a point set with polynomially bounded spread, the construction time improves from an $O_D$-bound to the respective $O_d$-bound. %

See Table~\ref{tab:related-work} for an overview of the size and runtime of known $\eps$-AVD constructions, and our proposed method.
Due to the large amount of previous work, we only include those methods that also compute cube-based Approximate Voronoi Diagrams in the comparison.
\ifx\WITHAPPENDIX\undefined\else
(Related work is discussed further in Appendix~\ref{sec:related-work}.)
\fi

\begin{table}[tb] \centering
	\setlength{\tabcolsep}{0.25em} %

	\begin{tabular*}{\linewidth}{c|c|ll} \hline
	Diagram & Technique                 & Size & ~~~Runtime  \\ \hline
	$\eps$-AVD   & Clustering, PLEB~\cite{Har-Peled01a}  &
	$\displaystyle O_D\!\left(n\frac{ \log n}{\eps^d} \log \frac{n}{\eps} \right)$   &
	$\displaystyle \times O_D\!\left(\log \frac{n}{\eps} \right)$
	\\
	$\eps$-AVD & Clustering, $\eps$-PLSB~\cite{SabharwalSS06} &
	$\displaystyle O_D\!\left(n\frac{\log 1/\eps}{\eps^{d+1}} \right)$    &
	$\displaystyle \times O_D\!\left(\log \frac{n}{\eps}\right)$
	\\
	$\eps$-AVD & Triangle ineq., $8$-WSPD~\cite[p148]{AryaM02} &
	$\displaystyle O_d\!\left(n\left(\frac{d}{\eps}\right)^d \log \frac{1}{\eps} \right)$ &
	$\displaystyle\times O_D\!\left(\frac{1}{\eps^d}\log \frac{n}{\eps}  \right)$
	\\
	$(1,\eps$)-AVD & Triangle ineq.~\cite[Cor.~9.10.f]{AryaMM09} &
	$\displaystyle O_D\!\left(n\frac{\log1/\eps}{\eps^{d-1}} \right)$ &
 \\ \hline
	$\eps$-AMWVD &
	Clustering, Sketches~\cite{Har-PeledK15}
	&
     $\displaystyle O_D\!\left( n \left( \frac{\log^{d+2}(n)} {\eps^{2d+2}} + \frac 1 {\eps^{d(d+1)}}\right) \right)$
	& ~\\
	$\eps$-AMWVD &
 Adaptive Refinement, $\eps^{-1}$-SSPD
	&
	$\displaystyle \Theta_d\!\left(n\frac{\log1/\eps}{\eps^{d-1}}\right)$ &
	$\displaystyle \times O_D\!\left( \frac{\log n}{\eps^{(d+5)/2}}  \right)$
	\\
	\end{tabular*}
	\caption{Overview of constructions of $\eps$-AVDs that provide {\em fast queries} for large $d$
 and the proposed method for $\eps$-AMWVDs.
	Note that $\eps$-AMWVDs are more general than the unweighted $\eps$-AVDs. %
	The time bound of~\cite{Har-PeledK15} is $O_D\!\left(n \log^{2d+3}(n)/\eps^{2d+2} + n/\eps^{d(d+1)} \right)$, and the query time
      $O(d \log (n/\eps^{d(d+1)}) )$ is \emph{cubic} in $d$.
    All other QuadTree based $\eps$-AVD methods have $O(d \log (n/\eps^d))$ query time.
	}\label{tab:related-work}
	\end{table}

\section{Preliminaries} \label{sec:cube-systems}

We provide a brief overview of canonical $d$-cubes and QuadTrees. 
The \emph{canonical cube system} is an hierarchical and infinite tiling of~$\mathbb R^d$ with canonical cubes. 
Level zero of the canonical cube system consists of unit cubes with vertices at integer coordinates. 
For all $\ell \leq -1$, we construct level $\ell$ by bisecting each cube in level $\ell+1$ along each of the $d$ axes. 
Therefore, there are $2^d$ cubes in level $\ell$ per cube in level $\ell+1$. 
For all $\ell \geq 1$, we merge $2^d$ cubes in level $\ell-1$ to obtain a single cube in level $\ell$, so that the cubes in level $\ell$ form a tiling of~$\mathbb R^d$. 
For example, a $d$-cube is a subset of points form $\R^d$ of the form $ {[2^\ell x_1, 2^\ell (x_1+1)]}  \times \ldots \times [2^\ell x_d, 2^\ell (x_d+1)]$ for integers $\ell, x_1, \ldots, x_d$.
Note that any two $d$-cubes from the system are either interior disjoint or one cube is a subset of the other.

Given a set of $n$ canonical $d$-cubes from the system, one can build a QuadTree on the set of cubes, in $O(d n \log \Delta)$ time, where $\Delta$ is the ratio between longest and shortest side length of the input set. 
In this work, we use compressed QuadTrees, which have $O(dn)$ size and can be constructed in $O(d n \log n)$ time. 
The subdivision of $\R^d$ induced by a QuadTree consists of canonical $d$-cubes, whereas the subdivision induced by a compressed QuadTree consists of regions that are the set difference of canonical $d$-cubes. %

\subsection{Voronoi Maps, Apollonian Balls, %
and the Core} 
\label{sec:def-setting}
Mapping $\lambda : \R^d \to \{1,\ldots,n\}$ is called a Voronoi Map for the distance functions $\{d_1,\ldots,d_n\}$, if 
$ d_{\lambda(x)}(x) \leq \min_i d_i(x)$,
for all points $x \in \R^d$.
The $d_i$ with index $i=\lambda(x)$ is called a nearest-neighbor of point~$x$.
In the case of Multiplicatively Weighted Voronoi Diagrams, each site $s_i \in \R^d$ has a positive weight-factor $w_i$ and the distance is $d_i(x)=\lVert x-s_i\rVert/w_i$.
We denote by $\lVert\cdot\rVert$ the Euclidean $\ell_2$-norm and indicate other $\ell_p$-norms explicitly by $\lVert\cdot\rVert_p$.
A subdivision of $\R^d$ is called MWVD if every cell in the subdivision is associated to one input site, and if mapping the points in a cell to the associated site is a Voronoi Map. 
Cell boundaries occur where the weighted distances to two sites are equal, which is along an Apollonian circle for $d=2$. 
For general~$d$, we define the Apollonian sphere between $s_i$ and $s_j$ to be 
$\{x \in \R^d: \lVert x-s_i\rVert/w_i = \lVert x-s_j\rVert/w_j\}$. 
A trivial MWVD is to construct the arrangement of the $\binom{n}{2}$ Apollonian spheres, giving a polynomial size bound.

\begin{figure}[tb]
    \centering
    \includegraphics[width=\textwidth]{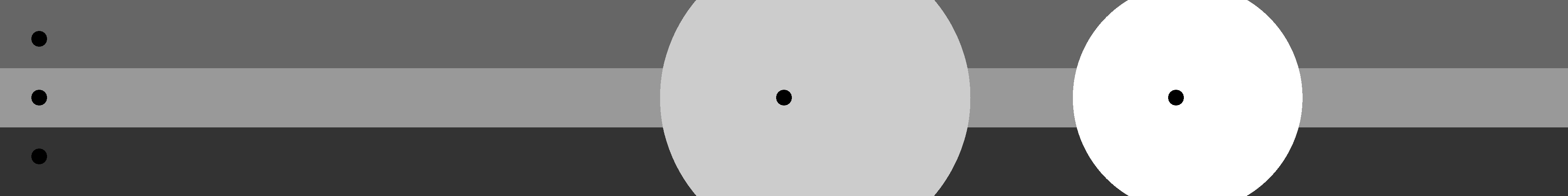}~\\
    \includegraphics[width=\textwidth,clip,trim=0 0 100 700]{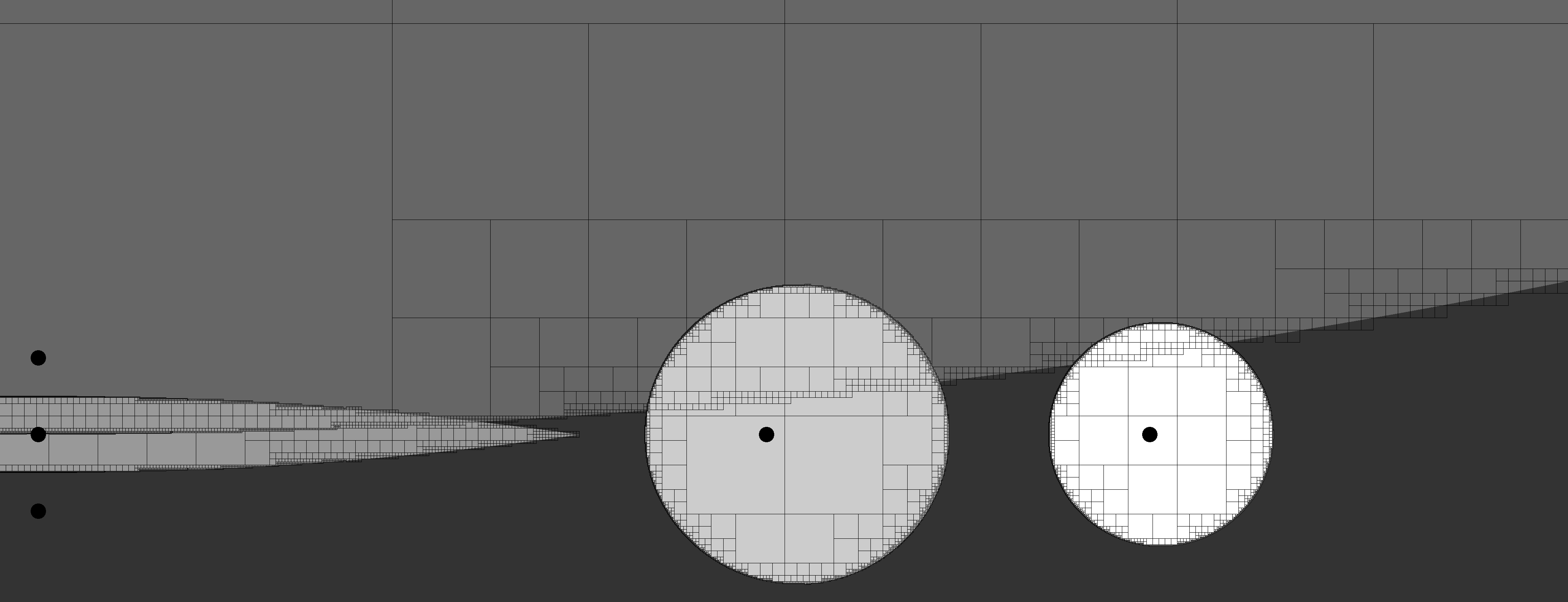}~\\
    \caption{The top shows an example of an exact MWVD of five sites ($\eps_S=0$).
     The bottom shows an $\eps_S$-AMWVD of the same instance obtained from cores with $\eps_S=0.01$.
     Result squares of the proposed \textsf{Adaptive Refinement} algorithm (Section~\ref{sec:core_algorithm}) for all four cores are shown as black overlay.}
    \label{fig:epss}
\end{figure}

\subparagraph*{Approximate Voronoi Maps of Apollonian Spheres and cube-based $\eps$-AVDs}
A mapping $\lambda : \R^d \to \{1,\ldots,n\}$ is called an $\eps$-approximate Voronoi Map for the functions $\{d_1,\ldots,d_n\}$, if 
$ d_{\lambda(x)}(x) \leq (1+\eps) \min_i d_i(x)$, for all points $x \in \R^d$.

Recall that the MWVD bisector of $s_j$ and $s_i$ is a $(d-1)$-dimensional hyper-plane, if $w_j=w_i$.
We introduce a parameter~$\eps_S \in [0,\eps)$, that we calibrate in Section~\ref{sec:sspd-proof}, and use it to $\eps_S$-approximate hyper-planes with hyper-spheres.
(This will turn out advantageous for obtaining optimal size.)
Let the sites be sorted by weight, so that $w_1 \leq \ldots \leq w_n$, breaking ties arbitrary but \emph{fixed}.
We define for all indices $i<j$ the Apollonian balls
\begin{align} 
    \disc(i,j) = \disc(s_i,s_j,\gamma_{ij}) 
               = \left\{x \in \R^d: \lVert x-s_i \rVert \gamma_{ij} \leq  \lVert x-s_j \rVert \right\}~,
\end{align}
where $\gamma_{ij} := \max(w_j/w_i,1+\eps_S)$. 
We call $\gamma_{ij}$ the {\bf \em effective weight} of $\disc(i,j)$.
For $\eps_S>0$, $\gamma_{ij} \geq 1+\eps_S$ and it follows that $\disc(i,j)$ is not a half-space. 
Note that the arrangement of the surfaces of all $\{\disc(i,j)\}$ yields an $\eps_S$-approximate Voronoi Map.
See \autoref{fig:epss}.

To enable \emph{fast} point location with Compressed Quad-Trees, an $\eps$-Approximate Voronoi Diagram ($\eps$-AVD) is a subdivision of $\R^d$ into $d$-cubes, and set-difference of $d$-cubes, that is an $\eps$-approximate Voronoi Map.
That is, each cube in the subdivision of $\R^d$ is associated to \emph{one} input site that is an $\eps$-Nearest-Neighbor for all points in the cube.

\subparagraph*{Closest, Furthest, and the Core of Apollonian Balls}
We further define $t^*(s_i,s_j,\gamma_{ij})$ to be the \emph{closest distance} from $s_i$ to a point on the surface of $\disc(i,j)$, and $t^\dagger(s_i,s_j,\gamma_{ij})$ to be the \emph{furthest distance} from $s_i$ to a point on the surface of $\disc(i,j)$.  
Note that these points are on the line through $s_i$ and $s_j$, and their distances are
\begin{align} %
    \gamma_{ij} &=\max(w_j/w_i,1+\eps_S)\\
    t_{ij}^* &= t^*(s_i,s_j,\gamma_{ij}) = \lVert s_j - s_i \rVert /(\gamma_{ij}+1)  \label{eq:def-t-star}\\
    t_{ij}^\dagger &= t^\dagger(s_i,s_j,\gamma_{ij})= \lVert s_j - s_i \rVert /(\gamma_{ij}-1)  \label{eq:def-t-dagger}~.
\end{align} 

For example, $\disc(i,j)$ has diameter $t^*_{ij}+t^\dagger_{ij}$.

Let the set of balls of site $s_i$ %
be $B_i := \left\{ (i,j)~:~ i<j~ \right\}$.
For every subset $A_i \subseteq B_i$, define the convex region
$core(A_i):= \bigcap_{(i,j)\in A_i}\disc(i,j)$.
By definition, the point $s_i \in core(A_i)$ for all non-empty $A_i\subseteq B_i$.

\section{Small Approximate Voronoi Diagrams using \texorpdfstring{$\binom{n}{2}$}{(n choose 2)} Bisectors}
\label{sec:core_algorithm}

The exact Voronoi region of site $s_j$ in an MWVD is 
$ core(B_j) \setminus \bigcup_{i<j} core(B_i)$ 
and a simple construction of the Voronoi Map may process the regions $core(B_j)$ by descending index $j$ and assign all points in $core(B_j)$ to the index $j$.
We introduce a suitable discretization for this idea next.

    \begin{restatable}{lemma}{coreisfat}
        \label{lem:core_is_fat}
        There exist two balls centered at $s_i$, one with radius $R$ containing $core(B_i)$, and one with radius $r$ contained in $core(B_i)$, so that $R/r \leq 3/\eps_S$.
        I.e.
        $core(B_i)$ is $3/\eps_S$-fat.
    \end{restatable}

\begin{proof}
Since any bisector has 
$t^\dagger_{ij}/t^*_{ij} = \frac{\gamma_{ij}+1}{\gamma_{ij}-1} \leq 1+2/\eps_S $
and the intersection of bisectors retains the maximum over those ratios, $core(B_i)$ is $3/\eps_S$-fat with $r:=\min_j\{t^*_{ij}\}$.
\end{proof}

To discretize a $\frac{R}{r}$-fat region for some $\eps_A \in (0,\eps_S)$, we consider the coarsest level where the canonical cubes have diameter at most $diam(C) \leq r\eps_A$, i.e. side-length $len(C)\leq r\eps_A/\sqrt{d}$.
Within distance at most $R$ from $s_i$, there are $O_d( ( \frac{2 R}{r} \cdot \frac{\sqrt{d}}{\eps_A} )^d )=O_d( (\sqrt{d}/\eps_A^2)^d)$ such cubes.
Checking each of the $k$ bisectors that define the fat region, we can determine with $O(k)$ distance computations if the centroid point of a cube is in $core(B_i)$.
Since any one cube is entirely inside, is entirely outside, or contains a point of the boundary, we have that only the latter case is potentially incorrect when deciding membership by the cube's centroid point.
Since any point $x$ on the boundary has $\lVert x-s_i\rVert \geq r$ and any point $q$ with erroneous membership decision has $\lVert q-x\rVert \leq \eps_A r$ from a point $x$ on the boundary (i.e. $d_i(x)=d_j(x)$), the discretization of the core approximates within a factor 
\begin{align*} \label{eq:ApxBoundRefinement}
 \frac{d_i(q)}{d_j(q)}
 &=
              \frac{d_i(q)}{d_i(x)} \cdot \frac{d_i(x)}{d_j(q)}
\leq       \left(1+\frac{\lVert x-q \rVert}{\lVert x-s_i \rVert}\right) \frac{d_j(x)}{d_j(q)}
\leq       \left(1+\frac{\lVert x-q \rVert}{\lVert x-s_i\rVert }\right) \left(1+\frac{\lVert x-q\rVert }{\lVert q-s_j\rVert }\right) \\
&\leq (1+\eps_A)(1+\frac{\eps_A r}{\lVert q-s_j\rVert }) 
 \leq (1+\eps_A)(1+\frac{\eps_A r}{\lVert x-s_j \rVert -\lVert x-q\rVert }) \\
&\leq (1+\eps_A)(1 + \frac{\eps_A r}{\lVert x-s_i\rVert -\lVert x-q\rVert})
 \leq (1+\eps_A)(1+\frac{\eps_A}{1-\eps_A})
  =1+O(\eps_A)~.
\end{align*}

\begin{observation}
    $O(1/\eps)$-fat cores allow a discretization of $O_d(n (\sqrt{d}/\eps^2)^{d})$ total size that $\eps$-approximates each core.
    Construction time is at most a factor $d\cdot n$ over the size bound.
    
\end{observation}

Note that the argument for cubes that intersect the boundary in our approximation bound already holds if the maximum distance of two points in a cube (diameter) is sufficiently small with respect to the distance to $s_i$, and not just if the diameter is at most $r\eps_A$.
Next, we discuss our, more space efficient, top-down search method that exploits this fact.
(Note that $O(\log(1/\eps_S))$ levels of the canonical cube system are relevant for any given core.)

As such, our \textsf{Adaptive Refinement} algorithm first determines
$r=\min\{t^*_{ij}\}$ from the given set of $k$ bisectors of site $s_i$, and then starts on the smallest canonical cube that contains the ball of radius $3r/\eps_S$ around the site $s_i$.
Recursively, we check if the current cube $C$ is entirely inside or entirely outside, i.e. 
$\lVert centr(C)-centr(ball(i,j))\rVert + diam(C)/2 \leq 
rad(\disc(i,j))$
for all $j>i$
or
$\lVert centr(C)-centr(ball(i,j))\rVert - diam(C)/2 > 
rad(\disc(i,j))
$ for a $j>i$.
If so, the search stops and includes the current cube $C$ in the result set, or respectively excludes it.
Otherwise, we check if the cube's diameter is sufficiently small for the centroid-test, i.e. 
$diam(C) \leq \eps_A (\lVert s_i- centr(C)\rVert - diam(C)/2)$.
If not, then all $2^d$ children of the cube are searched recursively.
If it is, then we stop the search and include the cube in the output set based on the result of its centroid-test, i.e. cube $C$ is included if and only if the centroid point of $C$ is inside \emph{each} of the $k$ bisectors that define $core(B_i)$.

Note that the search stops descending on a cube $C$ if one of the two criteria holds.
Termination and correctness follow immediately from the above discussion.
To improve on the above size bound, we bound the total number of canonical cubes that the search visits, each of which taking $O(d\cdot k)$ time. 

\begin{definition}[Distance Classes]
    Let $ball_s(x)=\{p:\lVert s-p\rVert\leq x\}$ be the ball of radius $x$ around site $s$.
    Let $L$ be the set of canonical cubes that our top-down search, \textsf{Adaptive Refinement}, visits.
    We partition $L=: \bigcup_j L_j$ in distance classes, such that $L_j$ contains those cubes $C \in L$ where $C \subseteq ball_s(2^j r)$ and $C \nsubseteq ball_s(2^{j-1} r)$.
\end{definition}

Note that $L_j = \emptyset$ for $j \leq -2$, since such a cube $C$ would be contained in $ball_s(r/4)$. Consequently, its parent $C'$ would be contained in $ball_s(r/2)$, satisfying the inclusion-test criteria that stops the search.
Thus, there are $O(\log(1/\eps_S))$ non-empty distance classes. %

We use Stirling's formula to
bound the volume of the Euclidean $d$-ball of radius $1$
\begin{equation} \label{eq:kappa-vol-bound}
    \kappa_d=\vol(ball(1)) %
    \in 
    \left[ 
        \frac{\pi^{d/2}}{\lceil d/2 \rceil!} ,
        \frac{\pi^{d/2}}{\lfloor d/2 \rfloor!} 
    \right]
    =\Theta_d(d^{-(d+1)/2}) ~.
\end{equation}

\begin{lemma}[Simple Bound]\label{lem:simple-bound}
    There are %
    $O_d(1/\eps_A^d)$
    canonical cubes in class $L_j$.
\end{lemma}

\begin{proof}
All cubes of distance class $L_j$ are contained in the $d$-ball around $s$ with radius $2^jr$, which has the volume  
$
            \vol(ball_s(2^j r)) 
     =       \kappa_d \cdot %
    (2^j r)^d 
    =     O_d( (2^jr)^d/d^{(d+1)/2})$.
Thus, it suffices to show that any one cube has side-length at least $\eps_A 2^j r/(8\sqrt{d})$.

 From the distance class partition, we have that a cube with diameter $\delta$ in class $j$ has that all of its points have distance $\geq 2^jr/2-\delta$ from the center $s$.

 Now, having the top-down search visit a cube $C$ with diameter $\delta$ would require the search did not terminate at its parent $C'$, which has diameter $2\delta$.
 Thus, $2\delta$ was not sufficiently small for stopping, i.e. 
 $2\delta > \eps_A (\lVert s_i -centr(C') \rVert-2\delta/2 )$.
 Since $centr(C')\in C$, its distance from $s_i$ is at least $2^jr/2-\delta$.
 Hence, $2\delta >\eps_A (2^jr/2-2\delta)$, which implies that $\delta > \frac{\eps_A}{1+\eps_A} \cdot 2^j r/4$.

Thus, any cube in $L_j$ must have diameter $\geq \eps_A \cdot 2^jr /8$, and consequently side-length $\geq \eps_A \cdot 2^jr/(8\sqrt{d})$.
\end{proof}

Thus, the lemma yields a running time bound and, consequently, a result size bound.
In \ifx\WITHAPPENDIX\undefined the full paper\else\autoref{sec:size-ub-is-wc-optimal}\fi, we show that this bound can be improved by one $(1/\eps_A)$-factor.

We summarize our results thus far before discussing how to assemble the Approximate Voronoi Diagram from the $\eps_A$-approximations of the cores.
\begin{theorem} \label{thm:new-adaptive-refine}
Let ${\cal R} \subseteq {\mathbb R}^d$ be a region that is the intersection of $k$ bisectors of $O(1/\eps_S)$-fatness, $s$ its center, and $\eps_A \in (0,\eps_S)$.
One can compute a set $L$ of 
$O_d\!\left( \log(1/\eps_S)/\eps_A^{d-1}\right)$ 
canonical cubes that $\eps_A$-approximates $({\cal R},s)$. 
Time is an $O(d\cdot k)$-factor over the size bound. %
\end{theorem}

Our lower bound in \autoref{lem:lower-bound-weighted} will show that $\Omega_d\!\left(\frac{\log(1/\eps)}{\eps^{d-1}}\right)$ cubes are required, if $\eps \ll 1/d^{3}$.

\subsection{Assembling the Approximate Diagram from Cubes}
\label{subsec:amwvd_from_approximate_cores}

In this section, we combine the $\eps_A$-approximations of each of the regions $core(B_i)$ to construct an $\eps$-AMWVD, where $\eps = (1+\eps_S)(1+\eps_A) - 1$. 
For each $1 \leq i < n$, we construct the $\eps_A$-approximate cubes for $(core(B_i),s_i)$ using \autoref{thm:new-adaptive-refine}. 
Each cube in the {$\eps_A$-ap\-prox\-i\-ma\-tion} of $(core(B_i),s_i)$ is given the label $i$. 
We collect all cubes for all labels $1 \leq i < n$ in a list. 
For $i=n$, we construct a canonical cube that contains all other canonical cubes for $1 \leq i < n$, and give this canonical cube the label $n$ and also add it to the list. 
(This cube will be at the root of the compressed QuadTree.)

Sort the list of canonical cubes by their $z$-order.
To remove duplicate cubes, iterate over the sorted list and keep only the cube with the minimum label (from those that are identical cubes).
Construct a compressed QuadTree from this set of canonical cubes using, say, the Divide\&Conquer approach (see Lemma~2.11 in~\cite{sariel-book}).
The leaves of the compressed QuadTree induces a subdivision of $\R^d$, where each cell in the subdivision is either a canonical cube, or the set difference of at most $2^d$ canonical cubes.

Finally, we label all cells in the compressed QuadTree as follows. 
The cubes that are from the the sorted list have their initial label, and the root has initial label $n$.
Starting at the root, if a child is unlabeled, or the child has larger label than its parent, then the child replaces its label with its parent's label. 
We repeat this process for all nodes in the compressed QuadTree in top-down fashion, say in a DFS traversal. 
This completes the construction of the approximate Voronoi Diagram.

To answer approximate (weighted) nearest-neighbor queries, given a query point $q \in \R^d$, we search our QuadTree for the smallest canonical cube containing $q$. 
The weighted nearest-neighbor of $q$ is the site with index equal to the label stored at this node.
Recall that point-location time in a compressed QuadTree is $O(d \log N)$ where $N$ is the number of cubes in the tree. %

Next, we prove the correctness of our proposed construction. 
When querying with a point $q$, we have two cases: Either the label returned is $n$, or it is less than $n$. 
If the label is $n$, then by construction, $q$ is in none of the $\eps_A$-approximations of $(core(B_i),s_i)$, for any $1 \leq i < n$. 
Therefore, $q$ is outside the $\eps_A$-approximation of $core(B_i)$ for all $1 \leq i < n$, so $s_n$ is indeed the site with the smallest weighted distance to $q$, up to a factor of $(1+\eps)$. 
Otherwise, let the label be $i$, for some $1 \leq i < n$. 
Due to the top-down propagation, we know that there is no canonical cube in the sorted list that both, contains $q$ and has label less than~$i$. 
Therefore, $q$ is outside the $\eps_A$-approximations of $core(B_j)$ for all $j < i$. 
So $q$ has smaller (weighted) distance to $s_i$ than any of $\{s_1,\ldots,s_i\}$, up to a factor of $(1+\eps)$. 
Moreover, we know that $q$ is in the $\eps_A$-approximation of $core(B_i)$. 
Therefore, up to a factor of $(1+\eps)$, $q$ has smaller weighted distance to $s_i$ than any of $\{s_i,\ldots,s_n\}$. 

Since the time for top-down label propagation is linear in the tree size, our construction time bound is one logarithmic factor over the size bound:

\begin{theorem}
\label{theorem:AMVWD}
Given $\eps_S>\eps_A >0$ and a set of balls $B_i$ for each $i < n$, one can compute an $\eps$-approximate Voronoi Diagram, where $\eps=(1+\eps_S)(1+\eps_A)-1$, with total size $O_d( n \log( 1/\eps_S)/\eps_A^{d-1} )$. 
The construction time is $O_d\!\left(\log \frac {n} {\eps_A}+ n^{-1} \sum_i |B_i|  \right)$ times the size bound. 
Moreover, time to locate a query-point is $O(d \log (n) + d^2 \log(1/\eps_A))$. 
\end{theorem}

This theorem will be used as a tool in \autoref{sec:sspd_section}, where we improve the construction time to near-linear, using our efficient construction of a bisector coreset for the $\{B_i\}$.
Note that the construction time is already quadratic in $n$, since $|B_i|<n$ for all $i$.
Next, we show that the result size is optimal, up to $\Theta_d(1)$ factors.

\section{A Matching Lower Bound for Diagram Size}\label{sec:lowerbound}

In this section, we show our matching lower bound for the size of $\eps$-AMWVDs.
That is, any subdivision comprised of axis-aligned hyper-rectangles requires  $\Omega_d(n \cdot \log(1/\eps)/\eps^{d-1})$ cells.
Our MWVD instances consist of $n$ copies of a two-site instance that are placed sufficiently far from each other.
The main idea for the two-site instance is that there are $\Omega(\log 1/\eps)$ distinct regions of space, each of which having a `large' total volume but having a geometric shape that only allows to cover a relative `small' volume with any one cell.
Though the basic approach is similar to the $\Omega_d(n\cdot \eps/(\sqrt{d}\eps)^d)$ lower bound in~\cite[Section~$5$]{AryaM02}, the difference is that that our argument addresses various sections of two Apollonian balls with curvatures $\Theta(\eps)$, instead of one hyper-cylinder that is bounded by two parallel hyper-planes.
This results in a bound that is stronger by a $(d^{(d-1)/2} \log \frac{1}{\eps})$-factor than the known bound for unit-weight $\eps$-AVDs, and matches our upper bound in \autoref{thm:new-adaptive-refine} up to $\Theta_d(1)$-factors.

Though it is an intriguing problem to also settle the question of optimal complexity for unit-weight $\eps$-AVDs, it is, unfortunately, quite unclear if one can obtain such a bound without curved MWVD bisectors.
(Cf. last two paragraphs of Section~8 in~\cite{AryaMM09}.)

\begin{restatable}{theorem}{lemlowerbound}
\label{lem:lower-bound-weighted}
Let $\varepsilon \in (0,1/16 d^3 ]$, $w_I=1$, $w_O=(1+\eps)^2$, and $B:=\disc(s_I,s_O,w_O/w_I)$ be the Apollonian ball of 
$s_I=(-1/\sqrt{d},\ldots,-1/\sqrt{d})$ and $s_O =((1+\eps)^2/\sqrt{d},\ldots,(1+\eps)^2/\sqrt{d})$.
Any subdivision of $\R^d$ in axis-aligned hyper-rectangles that is an $\eps$-approximation of the MWVD bisector $B$ must contain 
$\Omega_d\!\left( \log (1/\eps)/\eps^{d-1}\right)$ 
cells.
\end{restatable}

\begin{proof}
Any $\eps$-approximation of the MWVD of $B$ must assign the points inside $B_I:=\disc(s_I,s_O, {(1+\eps)^3})$ to site $s_I$ and outside $B_O:=\disc(s_I,s_O, 1+\eps)$ to site $s_O$, i.e. only the points in $B_O\setminus B_I$ may be labeled with either site.
Thus, any one cell $c$ in an $\eps$-approximation must not intersect both, $B_I$ and $\R^d \setminus B_O$.
Note that $B_I \subset B \subset B_O$ and the sites, as well as the centers $m_I$ and $m_O$, are co-linear, i.e. on the main diagonal.
From (\ref{eq:def-t-star}) and (\ref{eq:def-t-dagger}), we have that $t^*=1$ and that $t^*$ and $t^\dagger$ have the relations
\begin{align*}
t^\dagger_I (1+\eps)    &=  t^\dagger  = t^\dagger_O /(1+\eps)      \\
t^*_I (1+\eps)          &=  t^*        = t^*_O /(1+\eps)~,          \label{eq:lower-bound-t-star-relation}
\end{align*}
which shows that their radii, i.e. $r=(t^*+t^\dagger)/2$, have relation $r_I(1+\eps)=r=r_O/(1+\eps)$. 
The radii are $\Theta(1/\eps)$.

Let w.l.o.g.\! the $t^*_I$ point on $B_I$ be at the origin. 
Let $A$ contain the points from the upper half-space of $B_O \setminus B_I$, where upper/lower is due to a fixed hyper-plane that contains the main diagonal.
Define partition $A =: \cup_i A_i$ such that the points in $A_i$ have a norm in range $(2^i,2^{i+1}]$, and let $A_{-1}$ have the points with norm $\leq 1$.
We prove the following three claims \ifx\WITHAPPENDIX\undefined in the full version of the paper\else(see \autoref{sec:proofs-lb})\fi. 
\begin{restatable}{claim}{lemLBAmbPartsHaveLargeVol}
\label{lem:ambiguous-parts-have-large-volume}
    Let $A=B_O \setminus B_I$. %
    The $i$-th section $A_i=\{x \in A: \lVert x \rVert \in (2^i,2^{i+1}] \}$ has volume at least 
    $
    \vol(A_i) \geq  \eps 2^i\cdot \kappa_{d-1} 2^{(i+1)(d-1)-1} 
    = \Omega_d( \eps 2^{di}/d^{d/2})
    $.
\end{restatable}
\begin{restatable}{claim}{lemLBcellCoversSmallVolVI}
\label{lem:cell-covers-small-volume-vI}
    Any axis-aligned hyper-rectangle $c$, which does not contain points from $B_I$, can cover a volume of at most 
    $\vol(c \cap A_i) %
    = O_d((\eps2^i)^d/d^{(d+1)/2})$. 
\end{restatable}
\begin{restatable}{claim}{lemLBcellCoversSmallVolVO}
\label{lem:cell-covers-small-volume-vO}
        Let $\eps \in (0,1/d^3]$.
        Any axis-aligned hyper-rectangle $c$, which does not contain points from $\R^d\setminus B_O$, can cover a volume of at most 
    $\vol(c \cap A_i)  %
    = O_d((\eps2^i)^d/d^{(d+1)/2})$, provided index $i \leq \frac{5}{4}\log_2(1/\eps)$.
\end{restatable}

Thus, 
$
    \Omega_d\!\left( \frac{\eps 2^{di}/d^{d/2}}{(\eps2^i)^d/d^{(d+1)/2}} \right)
=   \Omega_d( \sqrt{d}/\eps^{d-1})
$ 
hyper-rectangles are necessary to cover any of the $\Omega(\log 1/\eps)$ many sections from $A$.
\end{proof}

\section{Approximate Cores: Computing Bisector Coresets Efficiently}
\label{sec:sspd_section}

Next, we define the notion $\eps$-approximation that we use for the proof (Section~\ref{sec:sspd-proof}) of the quality guarantee for the algorithm in Section~\ref{sec:sspd-algo}.
It extends the intuitive idea that `large balls' in the set $B_i$ may not be relevant for the intersection that defines $core(B_i)$.

Let $\adisc(i,j)$ denote the enlarged ball that is obtained by setting the effective weight to $w_j/ \alpha w_i$ in the bisector, i.e. $\adisc(i,j)=\disc(s_i,s_j,\gamma_{ij}/\alpha)$.
For $\alpha \geq 1$, we define a relation between any two subsets $X,Y \subseteq B_i$ from the bisectors of $s_i$ as
\begin{align*} 
X \prec_\alpha Y \quad \Longleftrightarrow \quad \forall ~(i,k) \in Y~:~  core(X) \subseteq \adisc(i,k) ~,
\end{align*}
and say for such a pair that $X$ is an $\alpha$-cover of $Y$.
Given a subset $X\subseteq B_i$, we call the largest %
subset $Y\subseteq B_i$ with $X \prec_\alpha Y$ the set of balls that are $\alpha$-covered by $X$.
Further, $X$ is called an $\alpha$-cover if it covers {\em all balls} in $B_i$, i.e. $X \prec_\alpha B_i$, and we have
\begin{align}
core(B_i) \quad \subseteq \quad  core(X) \quad \subseteq \quad \alpha\text{-}core(B_i) :=\bigcap_{(i,j) \in B_i} \adisc(i,j) \quad.
\end{align}

For example, the set of balls that are $1$-covered by a singleton set $\{(i,j)\}$ contains all balls $(i,k) \in B_i$ with $\disc(i,j) \subseteq \disc(i,k)$. 
Note that $X \prec_\alpha Y$ and $Y \prec_{\alpha'} Z$ implies $X \prec_{\alpha\cdot \alpha'} Z$. %
Clearly, using $\alpha$-covers $\{ A_1 , \ldots, A_{n-1}\}$ of the bisectors (i.e. $A_i \prec_\alpha B_i$ for all sites~$s_i$) turns the $\eps$-approximation algorithm of Section~\ref{subsec:amwvd_from_approximate_cores} into one that computes an
$\eps'$-approximate Voronoi Diagram, with $\eps'=(1+\eps)\alpha-1$, whose running time is sensitive to $|A_i|$.

The goal of our next algorithm is to compute a subsets $A_i \subseteq B_i$, so that $A_i$ is an $\alpha$-cover of $B_i$, and $A_i$ has \emph{constant size}. %
Then, we apply \autoref{theorem:AMVWD} to those bisector sets $\{A_i\}$.

\paragraph*{Recap: $\sigma$-Semi-Separated Pair Decompositions with Low Weight}%
Let $S \subseteq \R^d$ be a set of $n$ points.
A list of subset pairs 
${\cal P}=\{(X_i,Y_i) : X_i ,Y_i \subseteq S ,~ X_i \cap Y_i = \emptyset~\}$ is called a pair pecomposition if there is, for every $\{s,s'\} \in \binom{S}{2}$, a pair $(X_i,Y_i) \in {\cal P}$ with $|\{s,s'\}\cap X_i| = 1 = |Y_i \cap \{s,s'\}|$. 
The quantity $\sum_i (|X_i|+|Y_i|)$ is called the {\em weight} of the pair decomposition $\cal P$.
It is well known that a pair decomposition of $n$ points has weight $\Omega(n \log n)$.
(See~\cite[Lemma~3.31]{sariel-book}.)

A pair decomposition $\cal P$ of $S$ is called a $\sigma$-SSPD with respect to constant $\sigma>1$,
if every point set pair $(X,Y) \in {\cal P}$ has the separation property
\begin{align}   \label{eq:sspd-separation-diameter}
    \quad\quad
\min \Big\{~\max_{x,x'\in X}\lVert x-x' \rVert_2~,~\max_{y,y'\in Y}\lVert y-y' \rVert_2~\Big\} \cdot \sigma  \leq  \min_{x \in X, y \in Y} \lVert x-y\rVert_2~.
\end{align}
That is, the two sets have a closest-pair distance of at least $\sigma$ times the small diameter. %

Given a set of $n$ points from $\R^d$, a $\sigma$-SSPD with 
weight $w(n,d,\sigma)=O_d( d^{7d/2} \sigma^{d} n \log n)=O_D(\sigma^d n \log n)$  can be computed in deterministic
$O_D(\sigma^dn+n\log n)$ time~\cite[Theorem~$5$]{AbamBFGS11}.
For point sets with polynomially bounded spread, it is possible to improve both (deterministic) $O_D$-bounds to $O_d$-bounds with a QuadTree based pair decomposition, using~\cite[Lemma~2.8]{AbamH12}.

The efficiency of our coreset construction stems from low weight SSPDs.
We use the SSPD separation in terms of the radius of the two sets, which increases $\sigma$ by a factor of two.

\subsection{Computing Approximate Cores: SSPDs and Conic Space Partitions}\label{sec:sspd-algo}
Let $\beta, \eps_C>0$ and $\sigma \geq 2$ be constants, which we calibrate in Section~\ref{sec:sspd-proof}.
A $\beta$-cone around $s_i$ is an angular domain of the spherical coordinate system around $s_i$. 
Each of its $(d-1$) angular dimensions is partitioned into intervals of at most $2\beta$ radians.
For each $s_i$, we assign each $\beta$-cone a unique array index $j$, where $j = O_d(1/\beta^{d-1})$.
E.g. a rotation of at most $\beta$ radians suffices to rotate any point in the cone onto the cone's central ray. %

Let $\mathcal P$ be a $\sigma$-SSPD of the input sites $S$. For a pair $(L,H) \in \mathcal P$, we call $L$ the `light set' and $H$ the `heavy set' if $s_{\ell}$ is the site with maximum index in $L$, $s_h$ is the site with the maximum index in $H$, and $\ell < h$. 

Our algorithm maintains the following structure: 
For each site $s_i \in S$, and for each $\beta$-cone around $s_i$ with array index $j$, the data structure stores a set of partner sites $A_{ij}$.
Our algorithm populates the structure in three passes.
In our first pass, for each $(L,H) \in \mathcal P$, we reduce the size of $H$ to a subset $H'$. %
In our second pass, we iterate over $\mathcal P$ to initialize each of the sets $A_{ij}$.
Finally, the sets are populated in the third pass.

In our first pass, for each $(L,H) \in \mathcal P$, we construct a subset $H'$ of $H$. If the diameter of $H$ is at most the diameter of $L$, we set $H':=\{\}$.
If the diameter of $H$ is larger than the diameter of $L$, we construct $H'$ as follows.
Let $s_\ell \in L$ with $\ell$ maximal. For the $j^{th}$ cone around $s_\ell$, we let the sites of $H$ contained in this cone be $C_{\ell j}$. 
We use the following function:
\begin{lstlisting}[mathescape=true, backgroundcolor = \color{gray!10}]
$\textsf{Scan-Cone-Sites}(i,C,\eps_C)$:
    Let $C' := \emptyset$, $a = min\{ t^*_{ij} : s_j \in C\}$, and $b = min\{ t^\dagger_{ij} : s_j \in C\}$.
    Let $I_k=(x_k,x_{k+1}]$, with length $a \eps_C/2$ and $x_1=a$, cover $[a,b]$.
    Every interval $I_k$ holds one pointer.
    FOR $s_j \in C$ DO
        Compute the index $k$ with $t^*_{ij} \in I_k$.
        If diameter $(t_{ij}^*+t_{ij}^\dagger)$ is smaller than that of $I_k$'s reference,
            then set $I_k$'s pointer on $s_j$.
    FOR interval $I_k$  DO
        Add the kept bisector to result set $C'$.
    return $C'$\end{lstlisting}

We select for the $j^{th}$ cone a subset by setting $C'_{\ell j} := $~\textsf{Scan-Cone-Sites}$(\ell,C_{\ell j}, \eps_C)$ 
and define $H' = \cup_j C'_{\ell j}$. 
This completes the construction of $H'$. %

In our second pass, we initialize each cone of each site in our structure to store an interval~$[a,b]$.
We iterate over all pairs $(L,H) \in \mathcal P$ and all $s_i\in L \cup H$, and store for $j^{th}$ cone of~$s_i$, a variable $a$ equal to the minimum value of a~$t^*_{ik}$, and a variable $b$ equal to the minimum value of a~$t^\dagger_{ik}$. 
This minimum is taken over all sites $s_k \in H'\cup \{s_\ell,s_h\}$ that are in the $j^{th}$ cone of~$s_i$ and have $k>i$. This gives us the interval $[a,b]$.
After the pass over $\mathcal P$ is completed, we iterate over each cone of each site and partition the interval $[a,b]$ into disjoint intervals $I_k=(x_k,x_{k+1}]$ of length $a\eps_C/2$ that cover $[a,b]$, i.e. $x_{k+1}-x_k = a\eps_C/2$ and $x_1=a$.%
\ifx\WITHAPPENDIX\undefined\else~(See Appendix \autoref{fig:scan-cone}.)\fi

In our third pass, we populate the sets $A_{ij}$ based on the intervals $\{I_k\}$ of the $j^{th}$ cone of~$s_i$. We iterate over all pairs $(L,H) \in \mathcal P$ and maintain a reference from $I_k$ to the site that realized a minimum diameter.
For $s_i \in L \cup H$, and for the $j^{th}$ cone around $s_i$, we let the sites $s_m \in H'\cup\{s_\ell,s_h\}$ with $m>i$ that are contained in this cone be $C_{ij}$. 
For each $s_m \in C_{ij}$, we locate the interval $I_k$ of the cone that contains $t^*_{im}$ and compare the diameter of $\disc(i,m)$ with the smallest diameter of $I_k$ that we have encountered so far.
If the diameter of $\disc(i,m)$ is smaller, we set $s_m$ to be the site of $I_k$ realizing the minimum diameter. After the pass over all pairs is completed, for the $j^{th}$ cone of site $s_i$, and for all intervals $I_k$, we add the site that realized the minimum diameter for $I_k$ into the set $A_{ij}$.
This completes our three passes that construct the cone sets. 
Finally, we set $A_i = \cup_j A_{ij}$, and then apply \autoref{theorem:AMVWD} to the set of balls $A_i$. 

In the next section, we show that $A_i$ is an $\alpha$-cover of $B_i$. 
The algorithm's runtime bound 
$O_d( w(n,d,\sigma)\cdot m/\beta^{d-1} )$ 
will follow from weight $w(n,d,\sigma)$ of a $\sigma$-SSPD, the number of $\beta$-cones in the partitions of $\R^d$, and the maximum number $m$ of sites in the sets $A_{ij}$.

\subsection{Correctness: Choosing Sufficient \texorpdfstring{$\beta$, $\sigma$, and $\eps_C$}{β, σ, and εC}}
\label{sec:sspd-proof}
Our $(1+\eps)$ bound consists of seven components for each of the convex cores.
The components use the target approximation $\eps_A$ for the \textsf{Adaptive Refinement} in Section~\ref{sec:core_algorithm}, 
an $\eps_S$ that scales half-space bisectors to sufficiently large balls (see Section~\ref{sec:def-setting}),
an $\eps_C$ that is the tolerance for selecting a small set of sites per $\beta$-cone,
an $\eps_T$ that virtually translates sites along a ray from another site,
and $\eps_R$ that virtually rotate a site's partner (cf. \autoref{fig:four-cases}).

For prescribed $\eps>0$%
, we set the components such that 
\begin{align}
    (1+\eps_A)(1+\eps_S)(1+\eps_T)(1+\eps_R)^2(1+\eps_C)^2  ~&\leq ~ 1+\eps \\
    \max\{\eps_R,\eps_T,\eps_C\}~&<~\eps_S ~,
\end{align}
where the last inequality is {\em strict} to accommodate \autoref{lem:trans-heavy}.
For example, we can set $\eps_S=\eps/8$ and $\eps_A=\eps_C=\eps_R=\eps_T=\eps/16$.

This section shows 
$core(A_i) ~\subseteq~ \left(\frac{1+\eps}{1+\eps_A}\right)\text{-}core(B_i)$ and consequently the approximation bound of our approach.
Recall from Section~\ref{sec:def-setting} that all bisectors in $B_i$ have $w_j/w_i \geq 1+\eps_S$.

To show inclusion properties, we will use the following parametrization of balls in $B_i$:
Consider a fixed ray $q$, say the $x$-axis, that emanates from the origin $s_i$, having $w_i=1$.
Ignoring the input instance $S$ briefly, any pair $(s,w)$ of a point $s$ on $q$ and a real $w > 1$ defines a ball, with respective two points on the $x$-axis of distance $t^*, t^\dagger >0$.
It is convenient to use parametrization $(t^*,t^\dagger)$, instead of $(s,w)$, to describe this ball:
If input sites $s_j$ and $s_k$ are on the same ray $q$, then
~~$
    \disc(i,j) \subseteq \disc(i,k)
    ~~ \Leftrightarrow ~~
    t_{ij}^* \leq t_{ik}^* ~\wedge~ t_{ij}^\dagger \leq t_{ik}^\dagger
$~.
It is noteworthy that both inequalities can be decided without square-root computations (cf. Eq.~(\ref{eq:def-t-star}) and (\ref{eq:def-t-dagger})).

To show that every $(i,j) \in B_i \setminus A_i$ is $\alpha$-covered, the main idea is to consider the pair $(L,H) \in {\cal P}$ that separates it to observe that at least one bisector that $\alpha$-covers $(i,j)$ is contained in $A_i$. 
There are four cases for an absent bisector $(i,j)$:
(LL) $s_i \in L$ and $L$ has smaller diameter, 
(LH) $s_i \in L$ and $H$ has smaller diameter, 
(HL) $s_i \in H$ and $L$ has smaller diameter, and 
(HH) $s_i \in H$ and $H$ has smaller diameter. 
We use at most three affine transformations to bound each case.
See Figure~\ref{fig:four-cases}.
The bound for 
(LL) is $\alpha=(1+\eps_R)^2(1+\eps_C)^2(1+\eps_T)$, the bound for
(LH) is $\alpha=(1+\eps_R)^2(1+\eps_T)(1+\eps_C)$, the bound for
(HL) is $\alpha=(1+\eps_R)^2(1+\eps_T)(1+\eps_C)$, and the bound for
(HH) is $\alpha=(1+\eps_R)(1+\eps_C)$.              %

\begin{figure}[t]
    \includegraphics[width=\textwidth]{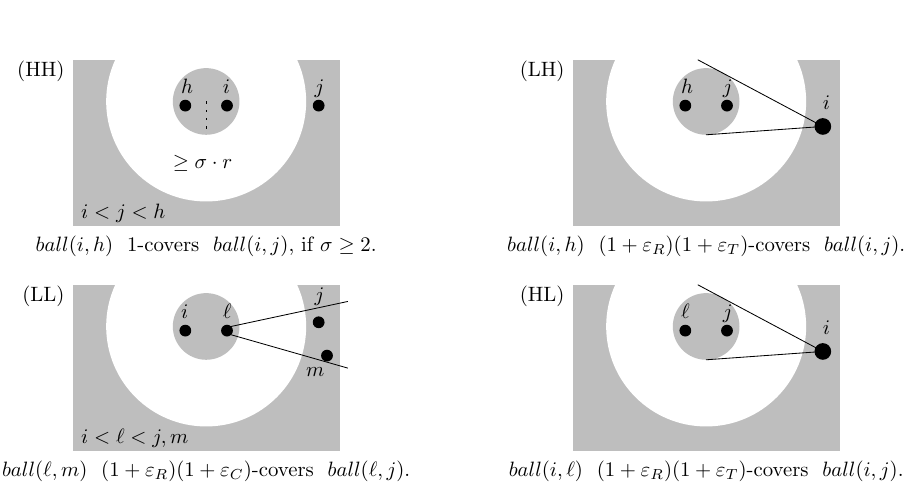}
    \caption{Cases (HH), (LH), (HL), and (LL), for covering an absent ball $(i,j) \in B_i\setminus A_i$.}
    \label{fig:four-cases}
\end{figure}

We start by showing an observation about pair decompositions.
A cluster of points $H$ that is, relative to its diameter, far from a given point $s_i$ can be rotated with a small angle onto a common ray $q$, from $s_i$ through an arbitrary point $s_h$ from the cluster.
\begin{observation}[Distant clusters]
    Let angle $\beta \in (0,1]$, $s_i \in S$, $c$ and $r$ be the center and radius of the minimum enclosing ball of $H \subseteq S\setminus \{s_i\}$, $\sigma := \lVert s_i - c \rVert/r >0$, and $s_h \in H$.
    If $\sigma \geq 2/\beta$, then $ \angle s's_is_h \in [0,\beta]$ for all $s' \in H$.
\end{observation}
\begin{proof} Since
    $ \frac{r}{r\sigma} = \tan \tfrac \beta 2 = \frac{\sin \beta}{1 + \cos \beta} \leq \frac{\beta}{1+(1-\beta^2)}=\frac{1}{2/\beta - \beta}$ and $\beta \geq 0$, any $\sigma \geq 2/\beta$ suffices.
\end{proof}

This observation motivates our main lemma that analyzes the enlargement of a ball from $B_i$ that is required to contain the ball that is obtained from a small rotation around $s_i$.

\begin{restatable}[Rotations at $s_i$]{lemma}{lemrotations}
\label{lem:rot}
If $\beta=\angle s_j s_i s_k \in [0,1]$ and
    $[t^*_{ij}, t^\dagger_{ij}]= [t^*_{ik},t^\dagger_{ik}]$, then

    $\disc(i,j)  \subseteq \adisc(i,k)$ 
for all $\alpha \geq 1 + \beta^2/2$.
\end{restatable}

Note that 
this bound also applies to rotations of $s_i$ on $s_j$ around $s_k$ for $k>i,j$, i.e. 
if $[t^*_{ik}, t^\dagger_{ik}]= [t^*_{jk},t^\dagger_{jk}]$ and  $\beta=\angle s_i s_k s_j$ is small, then $B \subseteq \adisc(j,k)$, where $B$ is the translation of $\disc(i,k)$ with the vector $ \overrightarrow{s_is_j}$. %

So far we showed that choosing a cone angle $\beta= \sqrt{2\eps_R}$ and $\sigma \geq 2/\sqrt{2\eps_R}$ satisfies the $(1+\eps_R)$-factors for rotations in all cases (i.e. LL, LH, HL, and HH).
Next we show that translations of sites in the low diameter set 
have a $(1+\eps_T)$-bound, for sufficiently large $\sigma$.

\begin{restatable}[Translations]{lemma}{lemtrans}
\label{lem:trans-heavy}
Let $p$ and $q$ be on a common ray from $s$, $\lVert s-p\rVert < \lVert s-q\rVert$, $\eps_T \in (0,\eps_S)$, point $m:=(p+q)\tfrac12$, 
$r:=\lVert m-p \rVert$.
If $1+\eps_T< \gamma$, then we have that 
$\lVert s -m\rVert\geq \sigma r $ implies that
    $
    t^*(s,q,\gamma)        ~\leq~ t^*\left(s,p,  \frac{\gamma}{1+\eps_T}\right)  
    $ and 
    $ 
    t^\dagger(s,q, \gamma)  ~\leq~ t^\dagger\left(s,p,\frac{\gamma}{1+\eps_T} \right) 
    $, 
for all $\sigma \geq 1+2/\eps_T$.
This also implies that
    $
     t^*(q,s,\gamma)        ~\leq~ t^*\left(p,s, \frac{\gamma}{1+\eps_T}\right)  
    $ and $
    t^\dagger(q,s, \gamma)  ~\leq~ t^\dagger\left(p,s,\frac{\gamma}{1+\eps_T} \right) 
    $.

\end{restatable}

The first translation property is used for the cases where $H$ has smaller diameter and the second for the cases where $L$ has smaller diameter.
One may think of the above discussion as a means to virtually place all sites in the low diameter set at the same spatial point with two transformations. %
We now show that partners of $s_i$ with lower weight than other partners, transformed to the same location, can be ignored in an $\alpha$-cover (e.g. \autoref{fig:four-cases} LH and HL). %

\begin{observation}[Weight Monotonicity]
    If $1< \gamma \leq \gamma'$, then 
    $ \disc(p,q,\gamma') \subseteq \disc(p,q,\gamma) $.
\end{observation}

\begin{proof}
    We give the, slightly more technical, argument for $t^\dagger(p,q,\gamma')\leq t^\dagger(p,q,\gamma)$.
    This holds iff
    $\frac{\lVert p-q\rVert}{\gamma'-1} \leq \frac{\lVert p-q\rVert}{\gamma-1}$, 
    which holds iff 
    $\gamma'-1 \geq \gamma-1$, since $\gamma-1\neq 0 \neq \gamma'-1$.
\end{proof}

Thus, for case (LH) and (HL) it suffices that $s_i$ scans $s_h$ and $s_\ell$, respectively.
(They are member of $H'\cup\{s_\ell\}\cup\{s_h\}$ and checked by algorithm when pair $(L,H)$ is considered.)
It remains to prove the $(1+\eps_C)$ factor in the approximations of \textsf{Scan-Cone-Sites}.

\begin{restatable}[Constant per cone]{lemma}{lemconstperray}
\label{lem:const-per-ray}
    Let $\{s_2,\ldots,s_n\}$ be on a common ray from $s_1$, $w_i/w_1\geq 1+\eps_S$, and $\eps_S>\eps_C >0$.
    Computing a $C_1 \subseteq B_1$ of size $O(1/\eps_C\eps_S)$, with $C_1\prec_{(1+\eps_C)}B_1$, takes $O(n)$ time.
\end{restatable}

    Thus, %
    selecting at most $m=O(1/\eps_C^2)$ sites per cone introduces only a factor of $(1+\eps_C)$.
    This completes the argument for all four cases, and we have
    $
    core(A_i) \subseteq \frac{1+\eps}{1+\eps_A}\text{-}core(B_i)
    $.
Taking $\sigma=1+2/\eps_T$, $\beta=\sqrt{2\eps_R}$,  and $m=O(\eps_C^{-2})$, the coreset construction time
$
    O_d(w(n,d,\sigma)\cdot m/\beta^{d-1})
=   O_D((\eps^{-d} n\log n) \cdot \eps^{-2} \eps^{-(d-1)/2} )
=   O_D(n \log(n)/\eps^{3(d+1)/2} )
$.
We summarize: %
\begin{theorem}\label{theorem:SSPD}
    The approximation algorithm computes, for each $1\leq i < n$, a subset $A_i \subseteq B_i$ with $core(A_i) \subseteq \frac{1+\eps}{1+\eps_A}$-$core(B_i)$ and $|A_i|=O_d(1/\eps^{(d+3)/2})$ in $O_D(n \log(n)/\eps^{3(d+1)/2} )$ time.
\end{theorem}
We are now ready to show our main result.

\begin{corollary}
For any $\eps>0$, one can compute an $\eps$-AMWVD of size $O_d(n \log(1/\eps)/\eps^{d-1})$.
The construction time is $O_D( \log(n)/\eps^{(d+5)/2} )$ times the output size.

The query time of the search structure is $O(d \log(n) + d^2\log(1/\eps))$.
\end{corollary}
\begin{proof}
Applying \autoref{theorem:AMVWD} on the bisector coresets that are obtained from \autoref{theorem:SSPD}, the construction time of the $\eps$-AMWVD is a factor %
$
O_d \left( |A_i| + \log(n/\eps) \right) 
=
O_d \left(  \log(n/\eps)/\eps^{(d+3)/2}  \right)
$
over the output size bound.
Hence, construction time is dominated by computing the bisector coreset, taking a factor 
$
O_D\left( \frac{n \log(n)/\eps^{3(d+1)/2}}{n \log(1/\eps)/\eps^{d-1}} \right)
=O_D(\eps^{-(d+5)/2}\log(n)/\log(1/\eps))
$ 
over the output size bound.
\end{proof}

\newpage

\bibliography{./refs/refs.bib}

\begin{thebibliography}{10}

\bibitem{AbamBFGS11}
Mohammad~Ali Abam, Mark de~Berg, Mohammad Farshi, Joachim Gudmundsson, and
  Michiel H.~M. Smid.
\newblock Geometric spanners for weighted point sets.
\newblock {\em Algorithmica}, 61(1):207--225, 2011.
\newblock \href {https://doi.org/10.1007/s00453-010-9465-2}
  {\path{doi:10.1007/s00453-010-9465-2}}.

\bibitem{AbamH12}
Mohammad~Ali Abam and Sariel Har{-}Peled.
\newblock New constructions of {SSPD}s and their applications.
\newblock {\em Comput. Geom.}, 45(5-6):200--214, 2012.
\newblock \href {https://doi.org/10.1016/j.comgeo.2011.12.003}
  {\path{doi:10.1016/j.comgeo.2011.12.003}}.

\bibitem{AbdelkaderAFM19}
Ahmed Abdelkader, Sunil Arya, Guilherme~Dias da~Fonseca, and David~M. Mount.
\newblock Approximate nearest neighbor searching with non-{E}uclidean and
  weighted distances.
\newblock In {\em Proc. 30th {ACM-SIAM} Symposium on Discrete Algorithms
  ({SODA}'19)}, pages 355--372, 2019.
\newblock \href {https://doi.org/10.1137/1.9781611975482.23}
  {\path{doi:10.1137/1.9781611975482.23}}.

\bibitem{AronovBK20}
Boris Aronov, Gali Bar{-}On, and Matthew~J. Katz.
\newblock Resolving {SINR} queries in a dynamic setting.
\newblock {\em {SIAM} J. Comput.}, 49(6):1271--1290, 2020.
\newblock \href {https://doi.org/10.1137/19M128733X}
  {\path{doi:10.1137/19M128733X}}.

\bibitem{AronovK22}
Boris Aronov and Matthew~J. Katz.
\newblock Dynamic approximate multiplicatively-weighted nearest neighbors.
\newblock In {\em Proc. 18th Scandinavian Symposium and Workshops on Algorithm
  Theory (SWAT'22)}, pages 11:1--11:14, 2022.
\newblock \href {https://doi.org/10.4230/LIPIcs.SWAT.2022.11}
  {\path{doi:10.4230/LIPIcs.SWAT.2022.11}}.

\bibitem{AryaAFM20}
Rahul Arya, Sunil Arya, Guilherme~Dias da~Fonseca, and David~M. Mount.
\newblock Optimal bound on the combinatorial complexity of approximating
  polytopes.
\newblock In {\em Proc. 32nd {ACM-SIAM} Symposium on Discrete Algorithms
  (SODA'20)}, pages 786--805, 2020.
\newblock \href {https://doi.org/10.1137/1.9781611975994.48}
  {\path{doi:10.1137/1.9781611975994.48}}.

\bibitem{AryaFM18}
Sunil Arya, Guilherme~Dias da~Fonseca, and David~M. Mount.
\newblock Approximate polytope membership queries.
\newblock {\em {SIAM} J. Comput.}, 47(1):1--51, 2018.
\newblock \href {https://doi.org/10.1137/16M1061096}
  {\path{doi:10.1137/16M1061096}}.

\bibitem{AryaM02}
Sunil Arya and Theocharis Malamatos.
\newblock Linear-size approximate {V}oronoi diagrams.
\newblock In {\em Proc. 13th {ACM-SIAM} Symposium on Discrete Algorithms
  (SODA'02)}, pages 147--155, 2002.
\newblock URL: \url{http://dl.acm.org/citation.cfm?id=545381.545400}.

\bibitem{AryaMM02}
Sunil Arya, Theocharis Malamatos, and David~M. Mount.
\newblock Space-efficient approximate {V}oronoi diagrams.
\newblock In {\em Proc. 34th {ACM} Symposium on Theory of Computing (STOC'02)},
  pages 721--730, 2002.
\newblock \href {https://doi.org/10.1145/509907.510011}
  {\path{doi:10.1145/509907.510011}}.

\bibitem{AryaMM09}
Sunil Arya, Theocharis Malamatos, and David~M. Mount.
\newblock Space-time tradeoffs for approximate nearest neighbor searching.
\newblock {\em J. {ACM}}, 57(1):1:1--1:54, 2009.
\newblock \href {https://doi.org/10.1145/1613676.1613677}
  {\path{doi:10.1145/1613676.1613677}}.

\bibitem{DBLP:journals/comgeo/AryaM00}
Sunil Arya and David~M. Mount.
\newblock Approximate range searching.
\newblock {\em Comput. Geom.}, 17(3-4):135--152, 2000.
\newblock \href {https://doi.org/10.1016/S0925-7721(00)00022-5}
  {\path{doi:10.1016/S0925-7721(00)00022-5}}.

\bibitem{AryaMNSW98}
Sunil Arya, David~M. Mount, Nathan~S. Netanyahu, Ruth Silverman, and Angela~Y.
  Wu.
\newblock An optimal algorithm for approximate nearest neighbor searching fixed
  dimensions.
\newblock {\em J. {ACM}}, 45(6):891--923, 1998.
\newblock \href {https://doi.org/10.1145/293347.293348}
  {\path{doi:10.1145/293347.293348}}.

\bibitem{AURENHAMMER1986}
Franz Aurenhammer.
\newblock The one-dimensional weighted {V}oronoi diagram.
\newblock {\em Information Processing Letters}, 22(3):119--123, 1986.
\newblock URL:
  \url{https://www.sciencedirect.com/science/article/pii/0020019086900554},
  \href {https://doi.org/https://doi.org/10.1016/0020-0190(86)90055-4}
  {\path{doi:https://doi.org/10.1016/0020-0190(86)90055-4}}.

\bibitem{AurenhammerE84}
Franz Aurenhammer and Herbert Edelsbrunner.
\newblock An optimal algorithm for constructing the weighted {V}oronoi diagram
  in the plane.
\newblock {\em Pattern Recognit.}, 17(2):251--257, 1984.
\newblock \href {https://doi.org/10.1016/0031-3203(84)90064-5}
  {\path{doi:10.1016/0031-3203(84)90064-5}}.

\bibitem{voronoi-book}
Franz Aurenhammer, Rolf Klein, and Der{-}Tsai Lee.
\newblock {\em {V}oronoi Diagrams and {D}elaunay Triangulations}.
\newblock World Scientific, 2013.
\newblock \href {https://doi.org/10.1142/8685} {\path{doi:10.1142/8685}}.

\bibitem{Chan02}
Timothy~M. Chan.
\newblock Closest-point problems simplified on the {RAM}.
\newblock In {\em Proc. 13th {ACM-SIAM} Symposium on Discrete Algorithms
  (SODA'02)}, pages 472--473, 2002.
\newblock URL: \url{http://dl.acm.org/citation.cfm?id=545381.545444}.

\bibitem{ChanHJ20}
Timothy~M. Chan, Sariel Har{-}Peled, and Mitchell Jones.
\newblock On locality-sensitive orderings and their applications.
\newblock {\em {SIAM} J. Comput.}, 49(3):583--600, 2020.
\newblock \href {https://doi.org/10.1137/19M1246493}
  {\path{doi:10.1137/19M1246493}}.

\bibitem{Chazelle86}
Bernard Chazelle.
\newblock Filtering search: {A} new approach to query-answering.
\newblock {\em {SIAM} J. Comput.}, 15(3):703--724, 1986.
\newblock \href {https://doi.org/10.1137/0215051} {\path{doi:10.1137/0215051}}.

\bibitem{dutch-book}
Mark de~Berg, Otfried Cheong, Marc~J. van Kreveld, and Mark~H. Overmars.
\newblock {\em Computational geometry: algorithms and applications, 3rd
  Edition}.
\newblock Springer, 2008.
\newblock URL: \url{https://www.worldcat.org/oclc/227584184}.

\bibitem{FanR20}
Chenglin Fan and Benjamin Raichel.
\newblock Linear expected complexity for directional and multiplicative
  {V}oronoi diagrams.
\newblock In {\em Proc. 28th European Symposium on Algorithms ({ESA}'20)},
  pages 45:1--45:18, 2020.
\newblock \href {https://doi.org/10.4230/LIPIcs.ESA.2020.45}
  {\path{doi:10.4230/LIPIcs.ESA.2020.45}}.

\bibitem{Gargantini82}
Irene Gargantini.
\newblock An effective way to represent quadtrees.
\newblock {\em Commun. {ACM}}, 25(12):905--910, 1982.
\newblock \href {https://doi.org/10.1145/358728.358741}
  {\path{doi:10.1145/358728.358741}}.

\bibitem{Har-Peled01a}
Sariel Har{-}Peled.
\newblock A replacement for {V}oronoi diagrams of near linear size.
\newblock In {\em Proc. 42nd Symposium on Foundations of Computer Science
  (FOCS'01)}, pages 94--103, 2001.
\newblock \href {https://doi.org/10.1109/SFCS.2001.959884}
  {\path{doi:10.1109/SFCS.2001.959884}}.

\bibitem{sariel-book}
Sariel Har-Peled.
\newblock {\em Geometric approximation algorithms}.
\newblock Number 173 in Mathematical Surveys and Monographs. American
  Mathematical Society, 2011.

\bibitem{Har-PeledIM12}
Sariel Har{-}Peled, Piotr Indyk, and Rajeev Motwani.
\newblock Approximate nearest neighbor: Towards removing the curse of
  dimensionality.
\newblock {\em Theory Comput.}, 8(1):321--350, 2012.
\newblock \href {https://doi.org/10.4086/toc.2012.v008a014}
  {\path{doi:10.4086/toc.2012.v008a014}}.

\bibitem{Har-PeledK15}
Sariel Har{-}Peled and Nirman Kumar.
\newblock Approximating minimization diagrams and generalized proximity search.
\newblock {\em {SIAM} J. Comput.}, 44(4):944--974, 2015.
\newblock \href {https://doi.org/10.1137/140959067}
  {\path{doi:10.1137/140959067}}.

\bibitem{Har-PeledR15}
Sariel Har{-}Peled and Benjamin Raichel.
\newblock On the complexity of randomly weighted multiplicative {V}oronoi
  diagrams.
\newblock {\em Discret. Comput. Geom.}, 53(3):547--568, 2015.
\newblock \href {https://doi.org/10.1007/s00454-015-9675-0}
  {\path{doi:10.1007/s00454-015-9675-0}}.

\bibitem{HeldL20}
Martin Held and Stefan de~Lorenzo.
\newblock An efficient, practical algorithm and implementation for computing
  multiplicatively weighted {V}oronoi diagrams.
\newblock In {\em Proc. 28th European Symposium on Algorithms ({ESA}'20)},
  pages 56:1--56:15, 2020.
\newblock \href {https://doi.org/10.4230/LIPIcs.ESA.2020.56}
  {\path{doi:10.4230/LIPIcs.ESA.2020.56}}.

\bibitem{Mount19}
David~M. Mount.
\newblock New directions in approximate nearest-neighbor searching.
\newblock In {\em Proc. 5th Conference on Algorithms and Discrete Applied
  Mathematics (CALDAM'19)}, pages 1--15, 2019.
\newblock \href {https://doi.org/10.1007/978-3-030-11509-8\_1}
  {\path{doi:10.1007/978-3-030-11509-8\_1}}.

\bibitem{SabharwalSS06}
Yogish Sabharwal, Nishant Sharma, and Sandeep Sen.
\newblock Nearest neighbors search using point location in balls with
  applications to approximate {V}oronoi decompositions.
\newblock {\em J. Comput. Syst. Sci.}, 72(6):955--977, 2006.
\newblock \href {https://doi.org/10.1016/j.jcss.2006.01.007}
  {\path{doi:10.1016/j.jcss.2006.01.007}}.

\end{thebibliography}

\ifx\WITHAPPENDIX\undefined
\else

\appendix

\newpage
\section*{Appendix}
%\nolinenumbers
\tableofcontents
%\linenumbers

\section{Related Work on $\eps$-AVDs, $\eps$-Nearest-Neighbor Search, and Generalizations.}
\label{sec:related-work}
For unweighted sites in $\R^d$, 
Arya et al.~\cite[Thm.1]{AryaMNSW98} showed that Approximate Nearest-Neighbor~(ANN) search can be solved with a balanced QuadTree-like structure of $O(n)$ cells %
with query time
$O_D(\log(n)/\eps^d)$.
Within the same space, Chan~\cite{Chan02} showed that point queries can be answered in $O_d(\log n + 1/\eps^{d-1})$ time on the RAM.
Har-Peled's seminal work~\cite{Har-Peled01a} introduced $\eps$-approximate Voronoi Diagrams~($\eps$-AVDs) and gave an efficient construction.
His `Point Location in Equal Balls'~(PLEB) framework uses a Minimum Spanning Tree for a hierarchical clustering of the sites, which is used to construct a compressed QuadTree of $O_D(n\frac{ \log n}{\eps^d} \log \frac{n}{\eps})$ size.
Its leaves form the approximate Voronoi decomposition, i.e.\!~each region is a $d$-cube or the set-difference of two $d$-cubes that is labeled with \emph{exactly one} of the input sites. 
Thus, point-location in the tree yields query time $O\!\left(d \log \left(n/\eps^d \right)\right)$, which is the \emph{best known query bound} that is attainable in near-linear size (e.g. across the dimension spectrum).
Sabharwal, Sharma and Sen~\cite{SabharwalSS06} refined the Point Location in Equal Balls framework to improve the size, and thereby construction time, bounds.

In the past two decades, an extensive work improved the size bounds in the $O_D$-regime.
Arya and Malamatos~\cite{AryaM02} showed that QuadTree based $\eps$-AVDs with $O(n(128d/\eps)^d\log\frac{1}{\eps})$  %
cells can be constructed, but also that 
$\Omega_d(n \eps /(\eps\sqrt{d})^d)$ 
cells are \emph{required} in any subdivision formed by axis-aligned hyper-rectangles.
Their $\eps$-AVD is constructed in two steps.
First, an $8$-Well-Separated Pair-Decoposition of the sites is used to derive a subdivision.
Each pair leads to $O((16d/\eps)^d\log \frac 1 \eps)$ cubes.
Their size and location is determined by two balls that are placed at the centroids of the two clusters (of the pair) whose radii are one quarter of the centroids' distance. 
For either centroid ball, $\Theta(\log \frac 1 \eps)$ concentric balls with doubling radii are added.
Second, each subdivision cell (from the QuadTree over all cubes) is then assigned to one site, the result of an $(\eps/4)$-ANN search for a point from the cell. 
This step, i.e. their Lemma~3.1, depends \emph{crucially} on the triangle inequality and fails for non-Euclidean distances\footnote{
	The $1$D MWVD of $s_1=0$, $s_2=1$, and $s_3=2$, with $w_1=w_2=1$ and $w_3\gg 1$, assigns $[-\frac{s_3}{w_3-1}, \frac{s_3}{w_3+1} ]$ to $s_1$.
	Since the smallest cube that contains $s_1$ has size $\Omega(\eps/4)$, assigning all of its points to $s_1$, though some beeing in the region of $s_3$, yields a relative error of $\Omega\!\left( \frac{\eps/w_1}{s_3/w_3}\right)=\Omega(\eps w_3)$, which is insufficient for an $\eps$-approximation.
	Thus, the lack of the triangle inequality poses a major difficulty for weighted sites.
}.
In~\cite{AryaMM09}, the authors refine this approach and show that a $5$-Well-Separated Pair-Decomposition suffices if one places in the subdivision, for each of the pairs, $O((20d/\eps)^d\log \frac{1}{\eps})$ many cubes (see proofs of their Lemmas~$6.2$ and $6.1$).
Further, they show that there is a set of $O_d( (d^d +1/\eps^{d-1})\log\frac{1}{\eps} )$ cubes for each pair of the $5$-WSPD (cf. proof of their Lemma~$9.4$).
Note that both size bounds contain $d^d$-factors.

For the {$\eps$-NN} search problem (point-queries) in the $O_D$-setting, the authors of~\cite{AryaMM02,AryaMM09} showed that $\eps$-dependency in the size can be reduced at price of an additive cost in the query bound (i.e. after point-location). 
Specifically for any tuning parameter $2 \leq \gamma \leq 1/\eps$, there is subdivision with $N=O_d(n( d^d+\gamma^{d-1})\log(1/\eps) )$ cells where each cell stores a list of $t=O_d(1/(\eps\gamma)^{(d-1)/2})$ `helper' sites that are checked against the query point, after the point-location, to determine the result.
`Spatial Amortization' shows that the total number of helpers remains $O_D(N)$ (see proofs of their Lemmas $8.2$ and $8.3$), thus yielding data structures with total size $O_D(N)$ and $O(d \log(N)+t)$ query time.
They also show that the size bound is tight up to $O_D(\log\frac{1}{\eps})$-factors when $\gamma=2$ and when $\gamma=1/\eps$.
In~\cite{AryaFM18}, the authors show that, using the lifting map, the site list can be organized as search DAG on MacBeath-ellipsoids, which allows to bound the additive term in the query cost by the \emph{product of the fan-out} and DAG depth, i.e. an $O(1)^d\cdot O(d\log\frac{1}{\eps})$ query-bound, while having a space-bound of $O_D(n/\eps^{d/2})$. %
In~\cite{AryaAFM20}, the authors improved the space bound of the involved polytopes to yield $\Theta_D(n/\eps^{(d-1)/2})$.
However, the query-bound \emph{grows exponential} with $d$.
Moreover, other query types, natural to Voronoi diagrams, are no longer provided in either of the $\eps$-NN search approaches for the $O_D$-setting.

For weighted sites in $\R^d$, 
the complexity of both problems, computing approximate diagrams and $\eps$-NN search, are not well understood.
The only result for approximate diagrams is due to Har-Peled and Kumar~\cite{Har-PeledK15}. 
They gave an efficient construction of Approximate Minimization Diagrams for $n$ input distance functions, given they are sufficiently well behaved.
Their Theorem~2.16 and Corollary~2.15 show that one can compute $\eps$-AMWVDs
of size $O_D\!\left(n \log^{d+2}(n) /\eps^{2d+2} + n/\eps^{d(d+1)} \right)$ 
in time $O_D\!\left(n \log^{2d+3}(n)/\eps^{2d+2} + n/\eps^{d(d+1)} \right)$,
where the query bound \emph{grows cubic} with $d$.
For the $\eps$-NN problem in the $O_D$-regime, the authors of~\cite{AbdelkaderAFM19} showed that Local Convexification can be used to extend their approximate polytope membership approach to non-Euclidean distance functions (not complying with the lifting map).
Their main result is that there is an $\eps$-NN structure of $O_D( n\log (1/\eps)/\eps^{d/2} )$ size that supports point queries in $O(d\log n)+ O(1)^d\cdot O(d \log(1/\eps))=O_d(\log \frac{n}{\eps})$ time %
(cf. Theorem~$1.1$ and p.$357$ in~\cite{AbdelkaderAFM19}; see also~\cite{Mount19}).
Note that
the query bound is \emph{exponential} in $d$.

Recently, Aronov and Katz~\cite[Section~$4$]{AronovK22} gave a novel approach for fully-dynamic $\eps$-NN search for point-queries in $O_d(\log^{d+1}(n)/\eps^{(d-1)/2})$ time, using $O_d((n/\eps^{(d-1)/2})\log^d n)$ space.

\newpage
\section{A Tight Bound on the Runtime of \textsf{Adaptive Refine}}
\label{sec:size-ub-is-wc-optimal}
\label{subsec:improving_approximate_core}

In this section we tighten the bound of \autoref{lem:simple-bound} for $|L_j|$ from 
$O_d(1/\eps_A^{d})$ to 
$O_d(1/\eps_A^{d-1})$. 
Recall that for $-1 \leq j \leq \lceil \log_2(3/\eps_S) \rceil$, class $L_j$ of consists of canonical cubes that are completely inside the ball of radius $2^j r$ centered at $s_i$, but are not completely inside the ball of radius $2^{j-1}r$ centered at $s_i$. 
We showed that all canonical cubes in level $j$ have side length $w \geq \eps_A \cdot 2^{j-3}r/\sqrt{d}$, which yields a bound of $|L_j|=O_d(1/\eps_A^d)$ many cubes in level $j$. 
Lemma~\ref{lem:surface_area_lemma} provides an alternative bound, drawing some inspiration from Lemma~3 in~\cite{DBLP:journals/comgeo/AryaM00}. 

\begin{restatable}{lemma}{surfacearealemma}
\label{lem:surface_area_lemma}
There are less than
$
    \sqrt{d} \kappa_d  5^d (2^jr/w)^{d-1}
= 
	O_d\!\left(\kappa_d (2^jr/w)^{d-1}\right)
$
canonical cubes of side-length $w$ in class $L_j$,
where $\kappa_d$ is the volume of the Euclidean $d$-ball of radius $1$.
\end{restatable}

\begin{proof}    
    Consider a canonical cube $c$ with side length $w$, let $c'$ be its parent canonical cube with side length $2w$, and let $C_j = ball_{s_i}(2^jr)$.
    Note that $c \in L_j$ yields $w \cdot \sqrt{d}\leq 2 \cdot r 2^j$.
    Let $\delta:=2w\sqrt{d}$ be the diameter of $c'$.

    To show the lemma's upper bound on the number of cubes with side-length $w$, we will show an upper bound on the number of their parent-cubes, i.e. the first number is within a $2^d$ factor of the second number.

    For \textsf{Adaptive Refine} to visit $c$, the search must not have terminated on $c'$.
    Meaning that $c'$ was not small enough for stopping with a centroid test, the minimum enclosing ball $MEB(c')$ was not entirely outside one bisector $b \in B_i$, and not entirely inside all bisectors 
    $MEB(c') \nsubseteq core(B_i)$.

    Let $D_{ij} := B_i \cup \{C_j\}$ additionally contain ball $C_j$.
    For a convex region $R \subseteq \R^d$, let $N(R,x):=\{p:\lVert p - p'\rVert \leq x \text{ for some }p'\in \partial R\}$ be the $x$-neighborhood of the boundary $\partial R$.
    From $c \subseteq C_j$ and the termination criteria, we have that parent $c'$ is contained in the region
    \begin{align}
        MEB(c') \quad \subseteq \quad 
        M_\delta 
        :=&  
            \left(\bigcap_{b \in D_{ij}}  (b \cup      N(b,\delta))    \right) 
        \setminus  \bigcap_{b \in D_{ij}} (b \setminus N(b,\delta))
        \quad.
    \end{align}

    Since $\vol(c')=(2w)^d$, it suffices to show that 
    $
    \vol( M_\delta ) = O_d( w \cdot \kappa_d \cdot (2^jr)^{d-1})
    $.

    We first upper bound the volume of the $\delta$-neighborhood $N(C_j,\delta)$ of the boundary $\partial C_j$.
    Recall that the volume of an $d$-dimensional annulus with inner-radius $A$ and width $B$ is 
    \begin{align*}
         \vol\left(ball(A+B) \setminus ball(A)\right)
    &= A^d \cdot \underbrace{\vol\left( ball(1) \right) }_{=\kappa_d} 
                ~\left[ (1+B/A)^d - (A/A)^d \right] 
     = A^d \kappa_d  \sum_{i=1}^d(B/A)^i  
    \\ = A^d \kappa_d (B/A) \sum_{i=0}^{d-1}(B/A)^i
     &= A^d \kappa_d (B/A) (1+B/A)^{d-1}            
     = B \kappa_d (A+B)^{d-1} ~,
    \end{align*}
    where the first equality is due to scaling each dimension by $A$.
    Thus,
    $$
            \vol\left(N(C_j,\delta)\right)
    =       2\delta \cdot \kappa_d \cdot (r2^j+\delta)^{d-1}
    \leq    4w\sqrt{d} \cdot \kappa_d \cdot (5r2^j)^{d-1}
    =       O_d(w \cdot \kappa_d \cdot (2^jr)^{d-1})
    $$ 
    from setting $A:=2^jr-\delta,~B:=2\delta$ and using 
    $
    \delta=2w\sqrt{d} \leq 2(r2^{j+1}/\sqrt{d})\sqrt{d}=4\cdot r 2^j$.

    It remains to show $\vol( M_\delta ) \leq \vol(N(C_j,\delta))$.\\
    Note that 
    $\vol(N(C_j,\delta))=\int_{-\delta}^{+\delta} \vol(S(C_j,x))\cdot dx$, where $S(C_j,x)=\partial ball_{s_i}(2^j+x)$ is the surface at distance $x$.
    Consider the partition of $M_\delta$ into the surfaces for $x \in [-\delta,\delta]$, i.e.
    $$
    M_\delta= 
    \bigcup_{x \in [0,\delta]} \underbrace{\partial \left(\bigcap_{b \in D_{ij}} b \cup N(b,x) \right)}_{=:~ S(D_{ij},x)}
    \quad \cup \quad
    \bigcup_{x \in (0,\delta]} \underbrace{\partial 
        \left(\bigcap_{b \in D_{ij}} b \setminus N(b,x) \right)}_{:=~S(D_{ij},-x)}
        \quad.
    $$

    Thus, it suffices to show that the surface areas satisfy
    $ \vol(S(D_{ij},x)) \leq \vol(S(C_j,x)) $. 
    See Figure~\ref{fig:cij}.
    Next, we use that $S(D_{ij}, x)$ is the surface of a convex region.

    \begin{figure}
        \centering
        \includegraphics{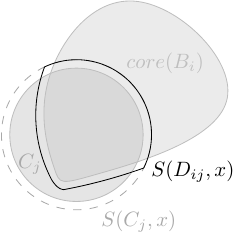}
        \caption{The core (gray), the ball $C_j$ of radius $2^jr$ centered at $s_i$ (gray), and the intersection region $D_{ij}$ of the core and the ball (dark gray). The set $S(D_{ij},x)$ denotes the set of points that are outside $D_{ij}$ and exactly distance $x$ away (black). The set $S(C_j,x)$ is defined similarly, for $C_j$.}
        \label{fig:cij}
    \end{figure}
    
    Consider the mapping $f: S(C_j, x) \to S(D_{ij}, x)$ that maps every point on $S(C_j, x)$ to its closest point on $S(D_{ij}, x)$. 
    Since $f$ is an orthogonal projection onto a closed convex set, we get that $f$ is a contraction mapping. 
    A contraction mapping is a mapping where $\lVert f(p) - f(q)\rVert \leq \lVert p-q\rVert$. 
    A consequence of $f$ being a contraction mapping is that the surface area of range of $f$ is at most the surface area of the domain of $f$. 
    Next, we show that $f$ is surjective, i.e. the range of $f$ is $S(D_{ij}, x)$. 
    Consider a separating hyperplane tangent to $D_{ij}$ at $p$. 
    Consider the ray from $p$, normal to the separating hyperplane, and pointing away from $D_{ij}$. 
    Let this ray intersect $S(C_j, x)$ at $q$. 
    Then the orthogonal projection of $q$ onto the separating hyperplane gives $p$. 
    Therefore, the closest point from $q$ to $D_{ij}$ is $p$, i.e. $f(q)=p$. Since $f$ is a surjective contraction mapping, the surface area 
    $\vol(S(D_{ij}, x)) \leq \vol(S(C_j, x))$.
\end{proof}

We use \autoref{lem:surface_area_lemma} to prove $|L| = O_d(\log(1/\eps_S)/\eps_A^{d-1})$. 
Let the side length of the smallest canonical cube in level $j$ be $w$. 
We showed in \autoref{lem:simple-bound} that $w \geq \eps_A \cdot 2^{j-3} r/\sqrt{d}$. 
All canonical cubes in level $j$ must have side length $2^k w$ for some non-negative integer~$k$. 
By \autoref{lem:surface_area_lemma}, for any~$k$, there are at most 
	$O_d\left( \kappa_d \cdot (2^{j-k}r/w)^{d-1} \right)$
canonical cubes of side length $2^kw$ in level~$j$. 
Summing over all $k\geq 0$ and taking the geometric series, there are at most 
$O_d\left( \kappa_d \cdot (2^j r/w)^{d-1} \right)$ 
canonical cubes of any side-length in level~$j$. 
Substituting $w \geq \eps_A \cdot 2^{j-3} r/\sqrt{d}$, we get that there are 
$O_d\left(\kappa_d \cdot (\sqrt{d}/\eps_A)^{d-1} \right)$ 
canonical cubes in level~$j$. 
Summing over the $O(\log 1/\eps_S )$ levels
and using 
$\kappa_d=O_d(d^{-(d+1)/2})$ from \autoref{eq:kappa-vol-bound}, 
we obtain 
$
    |L| 
=   O_d(\log(1/\eps_S)\frac{(\sqrt{d}/\eps_A)^{d-1}}{d^{(d+1)/2}})
=   O_d(\log(1/\eps_S)/\eps_A^{d-1})
$.

\newpage
\section{Proofs of the Size Lower Bound for AMWVDs}
\label{sec:proofs-lb}

\lemLBAmbPartsHaveLargeVol*
\begin{proof}%
Consider the surface-area of the hyper-sphere boundary of $B_I$ and the surface-area of the hyper-plane tangent on it in the origin, restricted to points with norm in $(2^i,2^{i+1}]$.
The origin is on the surface of ball $B_I$ and all hyper-spheres with radii $2^i$ centered there intersect $B_I$ for all $i$.
In the distance range $(2^i,2^{i+1}]$, the surface area of $B_I$ is larger than the surface area of the hyper-plane in this range.
(They are only equal for the sphere with infinite radius.)
The surface-area of the hyper-plane in this range is the volume of the annulus, with radii $2^{i+1}$ and $2^i$, in $\R^{d-1}$.
Its volume is 
$   \kappa_{d-1}\cdot ((2^{i+1})^{d-1} - (2^i)^{d-1}) 
\geq   \kappa_{d-1} 2^{(i+1)(d-1)}/2
=\Omega_d(\frac{2^{i(d-1)}}{d^{d/2}})$,
where the last bound uses \autoref{eq:kappa-vol-bound}.
To prove the $\vol(A_i)$ inequality, it therefore suffices to show that, for point $v_I$ on the $B_I$ surface of $A_i$, the closest point $v_O$ on the $B_O$ surface has distance $\geq \eps2^i$ for section $i\geq 0$.
Note that $v_O$ is unique since $v_I$ is not at the center $m_I$.
Let $\delta=\lVert v_I - v_O \rVert$ be the nearest-neighbor distance.
Let $\ell_I:=\lVert v_I-s_I\rVert$ and $\ell_O := \rVert v_O - s_O \rVert$, thus $\lVert v_I-s_O\rVert=\ell_I(1+\varepsilon)^3$ and $\lVert v_O-s_I\rVert=\ell_O(1+\eps)$.
For a contradiction, assume that $\delta < \eps \ell_I$.
In the case $\ell_I \leq \ell_O$, this yields

$$
\ell_O(1+\eps) = \lVert v_O - s_I\rVert \leq \ell_I + \delta < \ell_I(1+\eps) \leq \ell_O(1+\eps)~.
$$

In the case $\ell_I > \ell_O$, this yields %

$$
\ell_I(1+\varepsilon)^3 = \lVert v_I - s_O\rVert \leq \ell_O + \delta < \ell_I + \eps\ell_I = \ell_I(1+\eps)~,
$$

which contradicts the fact $(1+\eps)^3>1+3\eps$. 
Hence, $\delta\geq \eps\ell_I$ and $\ell_I > 2^i$, for $v_I \in A_i$. %
\end{proof}

%\newpage
\lemLBcellCoversSmallVolVI*
\begin{center}
\includegraphics[width=\textwidth,clip,trim=0 260 0 0]{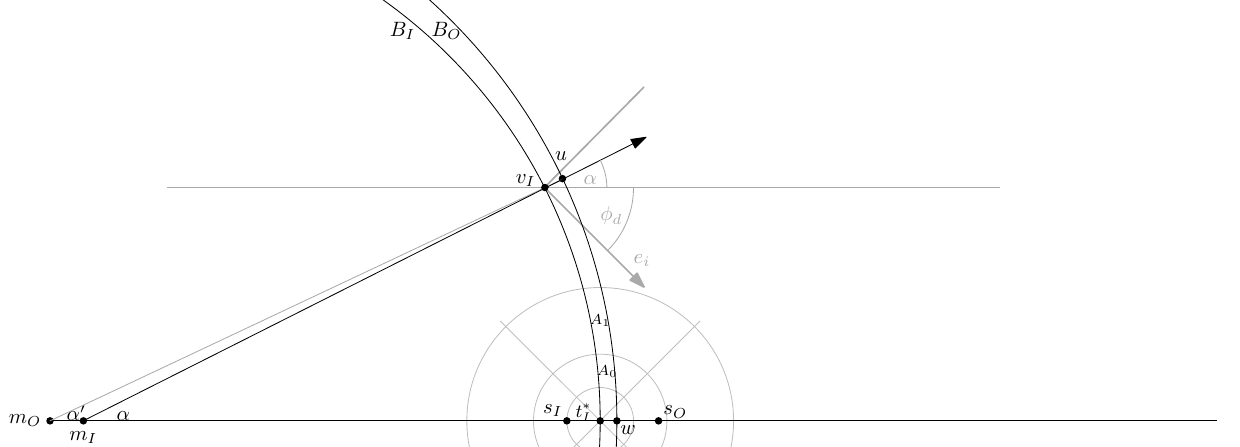}
\end{center}
\begin{proof}%
Let $v_I$ and $v_O$ denote the two corner points that span the hyper-rectangle~$c$.
We first give the argument for the case that corner $v_I \in A$.
In a cell with maximum coverage, i.e. $v_I$ is on the surface of $B_I$, we have that the $d$ edges, incident with corner $v_I$, end outside $B_O$.
Let $u$ be the intersection point of the ray, from $m_I$ through $v_I$, and the surface of $B_O$, and let $h=\lVert v_I - u\rVert$.
We first show that $h= O(\eps 2^i)$ for $\lVert v_I\rVert  \leq 2^{i+1}$.
Since $\{m_O,m_I,s_I,s_O\}$ are co-linear, we have 
$\alpha' := \measuredangle{u m_O s_I} \leq \measuredangle{v_I m_I s_I} =: \alpha$.
From the isosceles triangles in $m_I$ and in $m_O$, we have that
$\frac{\lVert v_I \rVert/2}{r_I}=\sin \frac{\alpha}{2}$ and 
$\frac{\lVert u - w\rVert/2}{r_O}=\sin \frac{\alpha'}{2}$,
where $w$ is the point on $B_O$ that has distance $t^*_O$ from $s_I$, i.e. $\lVert w \rVert= t^*_I( (1+\eps)^2 - 1) < 3\eps$.
Since $\sin \frac{\alpha'}{2} \leq \sin \frac \alpha 2$, we have
\begin{align*} 
            \lVert u \rVert 
   &\leq    \lVert w \rVert + \lVert u-w\rVert 
   \leq     \lVert w \rVert +  2r_O           \cdot \sin\frac{\alpha'}{2} 
   =        \lVert w \rVert +  2r_I(1+\eps)^2           \cdot \sin\frac{\alpha'}{2} 
   \\
   &\leq    \lVert w \rVert +  2r_I(1+\eps)^2 \cdot \frac{\lVert v_I \rVert}{2r_I}
   = \lVert w \rVert +  \lVert v_I \rVert(1+\eps)^2,
\end{align*}
and it follows that 
$
h 
=     \lVert u - v_I\rVert 
\leq  \lVert u \rVert - \lVert v_I\rVert 
\leq  \lVert v_I\rVert (1+\eps)^2 - \lVert v_I \rVert  + \lVert w \rVert 
< 3\eps \lVert v_I\rVert + 3\eps
$.
Next we show that $Vol(c\cap B_O) \leq (4h\sqrt{d})^d/d!$.
Consider the hyper-plane tangent to $B_O$ at $u$.
Its normal has angle $\alpha' \leq \alpha$ with the main diagonal $(1,\ldots,1)$.
The volume of $c\cap B_O$ is at most $\vol(c\cap H)$ where $H$ is the respective half-space of the hyper-plane.
To determine the volume of this $d$-simplex let w.l.o.g. $v_I$ be at the origin.
Its volume is 
$\frac{1}{d!}\det \left(e_1,\ldots,e_d \right)$, where the points $e_i$ span the simplex.
In other words, $e_i$ is only non-zero in the $i$-th component, where it contains the length of the boundary edge of $c$ (restricted to $H$) in dimension $i$.
It remains to show that $\lVert e_i \rVert = O(h\sqrt{d})$ for $\alpha \in \left[ 0, \frac{1}{4\sqrt{d}} \right]$.
Since $c$ is axis aligned, the angle between the main diagonal and $e_i$ is $\phi_d = \cos^{-1}(1/\sqrt{d}) \in [\pi/4,\pi/2)$.
Consider the problem in the 2D plane that contains the ray, from $m_I$ through $v_I$ and $u$, and the ray from the $i$-th edge that emanates from $v_I$.
In triangle $(v_I,u,e_i)$, 
the inner angle at $v_I$ is in the range $\phi_d \pm \alpha$ since $c$ is axis aligned, and 
the inner angle at $u$ is in the range $\pi/2 \pm \alpha'$ since the tangent's normal has angle $\alpha'$ with the main diagonal.
Thus, the angle at $e_i$ is at least 
$\frac{\pi}{2}-\phi_d -\alpha-\alpha'$, and the law of sines yields
\begin{align*}
      \lVert e_i \rVert 
&\leq  h \underbrace{\sin(\pi/2 + \alpha')}_{\leq 1}/\underbrace{\sin(\pi/2 - (\phi_d  +2\alpha))}_{=\cos(\phi_d +2\alpha)}    
\leq  h / \cos(\phi_d +2\alpha)
\\
&=  h\sqrt{d} / (\cos 2\alpha - \sqrt{d-1}\sin 2\alpha) 
\\
&\leq  h\sqrt{d} / (1-(2\alpha)^2 - \sqrt{d-1} (2\alpha)) 
\quad
\leq  4 h \sqrt{d}
  &&\forall~\alpha \leq 1/4\sqrt{d-1}~,
\end{align*}
where the second last inequality uses $\cos x \geq 1-x^2/2$ and $\sin x \leq x$.
Thus, $\lVert e_i \rVert = O(h\sqrt{d})$ and using that
$d! > \sqrt{2\pi d}~(d/e)^d$ for all $d\geq 1$
gives $Vol(c \cap B_O)\leq (4eh)^d/d^{(d+1)/2}$.
\end{proof}

%\newpage
\lemLBcellCoversSmallVolVO*
\begin{center}
\includegraphics[width=\textwidth,clip,trim=0 260 0 0]{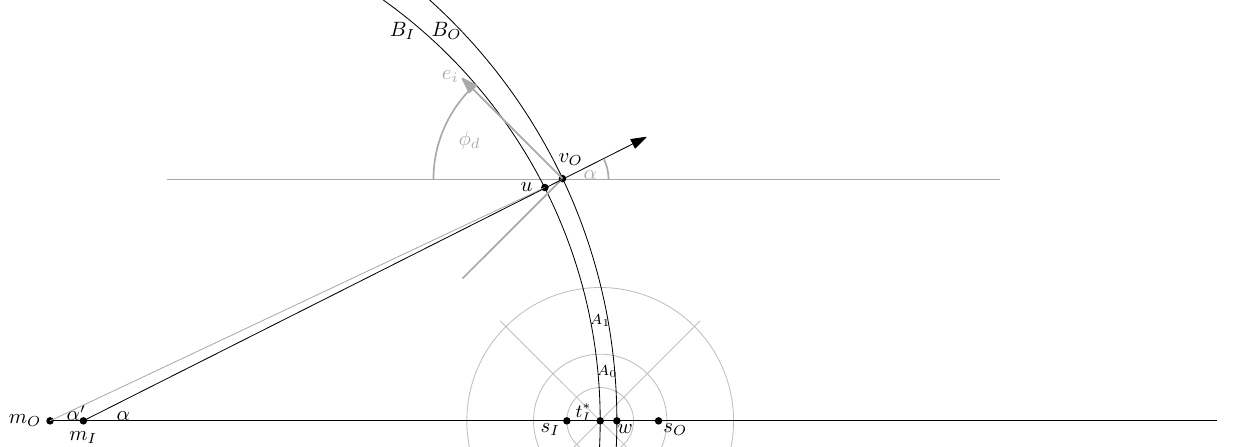}
\end{center}
\begin{proof}
For the case $v_O \in A$, a cell with maximum coverage has $v_O$ on the surface of $B_O$ and the $d$ edges incident with corner $v_O$ end inside $B_I$.

Let $u$ be the intersection point of the ray, from $m_I$ through $v_O$, and the surface of $B_I$, and let $h=\lVert v_O - u\rVert$.
We first show that $h= O(\eps 2^i)$ for $\lVert v_O\rVert  \leq 2^{i+1}$.
Since $\{m_O,m_I,s_I,s_O\}$ are co-linear, we have 
$\alpha' := \measuredangle{v_O m_O s_I} \leq \measuredangle{v_O m_I s_I} =: \alpha$.
From the isosceles triangles in $m_I$ and in $m_O$, we have that
$\frac{\lVert u \rVert/2}{r_I}=\sin \frac{\alpha}{2}$ and 
$\frac{\lVert v_O - w\rVert/2}{r_O}=\sin \frac{\alpha'}{2}$,
where $w$ is the point on $B_O$ that has distance $t^*_O$ from $s_I$, i.e. $\lVert w \rVert= t^*_I( (1+\eps)^2 - 1) < 3\eps$.
Since $\sin \frac{\alpha'}{2} \leq \sin \frac \alpha 2$, we have
\begin{align*} 
            \lVert v_O \rVert 
   &\leq    \lVert w \rVert + \lVert v_O-w\rVert 
   \leq     \lVert w \rVert +  2r_O           \cdot \sin\frac{\alpha'}{2}
   =        \lVert w \rVert +  2r_I(1+\eps)^2           \cdot \sin\frac{\alpha'}{2}
   \\
   &\leq     \lVert w \rVert +  2 r_I(1+\eps)^2 \cdot \frac{\lVert u \rVert}{2r_I}
   =        \lVert w \rVert +  \lVert u \rVert(1+\eps)^2,
\end{align*}
and it follows that 
$
        h 
=       \lVert v_O-u \rVert 
\leq    \lVert v_O \rVert - \lVert u\rVert 
\leq  
        \lVert u\rVert (1+\eps)^2 - \lVert u \rVert  + \lVert w \rVert 
<       3\eps \lVert u\rVert + 3\eps 
$.

To construct a bounding simplex analogue to the proof of \autoref{lem:cell-covers-small-volume-vI},
we consider the problem in the 2D plane that contains the ray from $m_I$, and the ray from the corner $v_O$ in the $i$-th axis aligned direction.
Let, without loss of generality, $v_O$ be at the origin.
The angle at $v_O$ is again in the range $\phi_d\pm \alpha$, and we use the triangle $(m_I,v_O,e_i)$, where $e_i$ is the (first) intersection point %
of the edge-ray with the circle of radius $r_I$ around $m_I$.
Note that the simplex spanned by the corner points $e_i$ contains the volume $c \cap A$.
Since $\lVert m_I \rVert = r_I+h$ and $\lVert e_i - m_I\rVert = r_I$, we have from the cosine law the quadratic equation

$$
0 = \lVert e_i \rVert^2 -2(r_I+h)\lVert e_i \rVert \cos (\angle v_O) +(r_I+h)^2 -r_I^2~,
$$
where $\cos(\angle v_O)>0$.
Since it suffices to bound $\lVert e_i \rVert$ in terms of $h$, we seek an upper bound for the solution

$$
\lVert e_i \rVert = \cos (\angle v_O) (h+r_I) - \sqrt{ \Big(\cos (\angle v_O)(h+r_I)\Big)^2 -h(h+2r_I)}~.
$$
Using the Laurent-Expansion %
$\sqrt{x^2-y}=x-\frac{y}{2x}-\frac{y^2}{8x^3}+O(1/x^4)$, we have
\begin{align*}
\lVert e_i \rVert 
&\leq 
\underbrace{\frac{h(h+2r_I)}{2\cos(\angle v_O)(h+r_I)}}_{\leq 3 h \sqrt{d}} + 
\underbrace{\frac{h^2(h+2r_I)^2}{8\cos^3(\angle v_O)(h+r_I)^3}}_{\leq 9h^2 d^{3/2}/r_I} 
\quad \leq \quad 3h\sqrt{d}+O(1)\quad
,
\end{align*} 
where the last row used the fact $h < r_I$ and $\cos \angle v_O \geq  \cos(\phi_d + \alpha)>1/2\sqrt{d}$.
Since $r_I=\Theta(1/\eps)$, it suffices to show that both $h^2$ and $d^{3/2}$ are at most $O(1/\sqrt{\eps})$, where the latter clearly holds for all $\eps \leq 1/d^3$.
Since $h=O(\eps 2^i)$, we have that $h^2 =O(1/\sqrt{\eps})$ for all $i\leq \log_4(1/\eps^{5/2})$. %

Recall that $d! > \sqrt{2\pi d}~(d/e)^d$ for all $d\geq 1$.
Since $\lVert e_i \rVert = O(h\sqrt{d})$ for every corner $e_i$ of the simplex, we showed 
that any one hyper-rectangle $c$ can cover at most 
$
        Vol(c\cap B_O) 
=     O(h\sqrt{d})^d/d! 
=       O_d((h/\sqrt{d})^d)
$
volume.
\end{proof}

\newpage
\section{Proofs of the Approximation Bound of our Coreset Construction}

\lemrotations*

\begin{figure}[H]\centering
    \includegraphics[height=5cm]{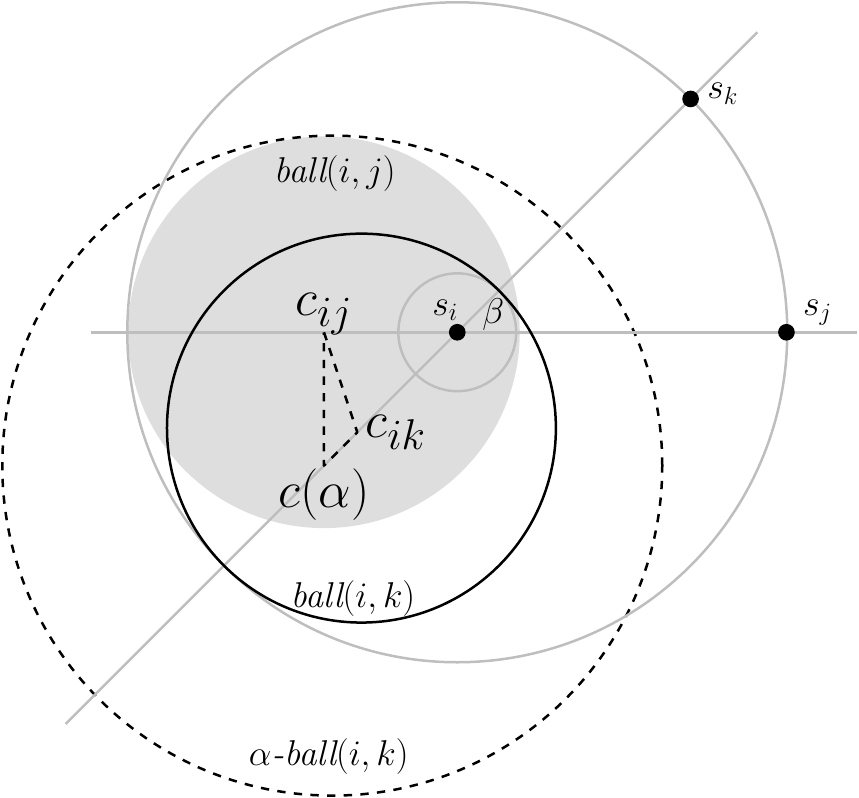}
    \caption{Illustration of the proof of Lemma~\ref{lem:rot}.}
\end{figure}

\begin{proof}
    Since $\lVert s_j \rVert=\lVert s_k \rVert < \infty$ and $\gamma_{ij}=\gamma_{ik}>1$, we have $t_{ij}^\dagger, t_{ik}^\dagger<\infty$.
    Let $s_i$ be at the origin and $\gamma=\gamma_{ij}=\gamma_{ik}$.
    
    Define $r(s_i,s_j,\gamma_{ij})$ and $c(s_i,s_j,\gamma_{ij})$ to be the radius and center of $\disc(i,j)$. Since the diameter is $t^*+t^\dagger$, we have for the radius and center the equations:
    \begin{align}
        r(s,s',\gamma) &= \frac{\lVert s - s'\rVert}{\gamma- 1/\gamma} \\
        c(s,s',\gamma) 
        &=s - (s'-s)/(\gamma^2-1)~.
    \end{align}
    
    The norm $c(\alpha)=c(s_i,s_k,\gamma/\alpha)$ increases monotonically with $\alpha\geq 1$ until the angle $\angle c(\alpha)c_{ij}s_i=\tfrac\pi 2$.
    We have $\disc(i,j)\subseteq \adisc(i,k)$ for any $\alpha$ of at least this value.
    Consider the triangle $c(\alpha)s_ic_{i,j}$ and let
    \begin{align*}
    x  &:= \lVert c_{ij} \rVert 
       = r_{ij}-t^*_{ij} 
        = \frac{\Vert s_j\rVert}{\gamma-1/\gamma} - \frac{\Vert s_j\rVert}{\gamma+1}
        = \Vert s_j\rVert \frac{\gamma+1}{\gamma^3 + \gamma^2 -\gamma - 1}
        = \frac{\Vert s_j\rVert}{\gamma^2 -1 }
    \\
    y(\alpha) &:= \lVert c(s_i,s_k, \gamma/\alpha) \rVert 
              = r(s_i,s_k,\gamma/\alpha)-t^*(s_i,s_k,\gamma/\alpha)
               = \frac{\Vert s_k\rVert}{(\gamma/\alpha)^2 -1}
    \end{align*}
    denote the two side lengths therein.
    Thus, we have at the smallest admissible $\alpha \geq 1$ that
    \begin{align}
    \cos \beta &= \frac{x}{y(\alpha)} 
               = \frac{\lVert s_j \rVert }{\lVert s_k \rVert} \frac{(\gamma/\alpha)^2-1}{\gamma^2-1} 
    \quad\quad
    \leadsto\quad
    \alpha^2 =
     \frac{\gamma^2}{ 1+ \frac{\lVert s_k \rVert }{\lVert s_j \rVert}(\gamma^2-1)\cos \beta }~.
    \end{align}
    
    Since $\lVert s_j \rVert=\lVert s_k \rVert$, setting $\eps=\alpha-1$ and $z=\gamma^2-1$, we obtain the sufficient condition
    $$
        1+2\eps+\eps^2 
    =
        \alpha^2
    \geq   \frac{1+z}{ 1+ z\cos \beta }
    =   1 + \frac{z-z\cos \beta}{ 1+ z\cos \beta }~,
    $$
    which is satisfied trivially if $z=0$ or $\beta=0$.
    For $z>0$, we use the fact that $\cos \beta > 1-\beta^2/2$ for all $\beta \neq 0$
    and thus have
    $$
            \frac{z-z\cos \beta}{ 1+ z\cos \beta }
    =       \frac{1-\cos \beta}{ \underbrace{1/z}_{>0} + \cos \beta }
    \leq       \frac{\beta^2/2}{ 1-\beta^2/2 }
    =    \frac{\beta^2}{ 2-\beta^2 }
    \leq    \beta^2.
    $$
    Since $1+2\eps+\eps^2\geq 1+ 2\eps$, having $\eps = \beta^2/2$ suffices.
    \end{proof}

\lemtrans*
\begin{figure}[H] \centering
    \includegraphics[scale=.8]{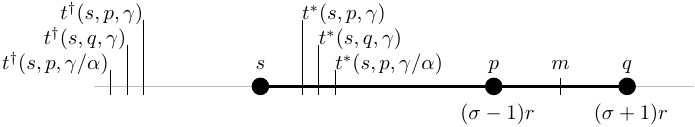}
\caption{Illustration of the proof of Lemma~\ref{lem:trans-heavy}.}
\end{figure}

\begin{proof} 
Since $\lVert x-y\rVert = \lVert y-x\rVert$, it suffices to prove the first property.
Let $\alpha:=1+\eps_T$ and $s$ at the origin.

We start by showing $t^*(s,q,\gamma) \leq t^*(s,p,\gamma/\alpha)$, which holds if and only if
\begin{align*}
    \frac{(\sigma+1)r}{\gamma+1}    &\leq       \frac{(\sigma-1)r}{\gamma/\alpha+1} \\
    1 + \frac{2}{\sigma-1}=\quad\quad
    \frac{\sigma+1}{\sigma-1}       &\leq       \frac{\gamma+1}{\gamma/\alpha+1} 
    \quad\quad =
    1+\frac{\gamma-\gamma/\alpha}{\gamma/\alpha+1} = 1+\underbrace{\frac{1-1/\alpha}{1/\alpha+1/\gamma}}_{\geq \frac{1-1/\alpha}{1/\alpha}=\alpha-1}~.
 \end{align*}
Thus it suffices to have $1 + \frac{2}{\sigma-1} \leq \alpha$, which is true for all 
$\sigma\geq 1+\frac{2}{\alpha-1}$.

We have $t^\dagger(s,q,\gamma) \leq t^\dagger(s,p,\gamma/\alpha)$ if and only if
\begin{align*}
1 + \frac{2}{\sigma-1}=\quad\quad
\frac{\sigma+1}{\sigma-1}       &\leq       \frac{\gamma-1}{\gamma/\alpha-1} 
\quad\quad = 
1+\frac{\alpha-1}{1-\alpha/\gamma}
=1+\underbrace{\frac{\eps_T}{ 1-(1+\eps_T)/\gamma}}_{=:x}
\end{align*}
Now, $x>0$ since $\gamma>1+\eps_T$.
Moreover, $x\geq \eps_T$ since $1-\frac{1+\eps_T}{\gamma}\in(0,1)$.\\
Hence, $\sigma \geq 1+2/\eps_T$ suffices for both inequalities.
\end{proof}

\clearpage
\begin{figure}[H]\centering

            \includegraphics[height=3cm]{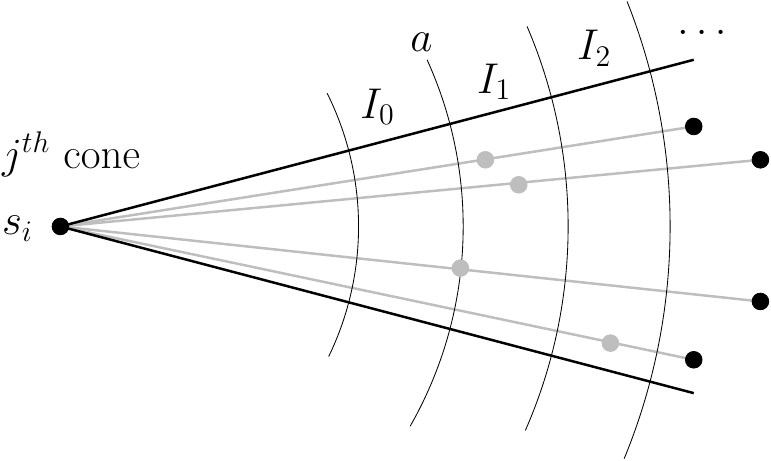}
            
        \caption{ %
        The values $t^*_{ij}$ in $[a,b]$ are partitioned by intervals $I_0,\ldots,I_m$ of length $a \eps_C/2$.}
        \label{fig:scan-cone}
\end{figure}
\lemconstperray*
    \begin{proof}
        We clarify the notation by revisiting the algorithm \textsf{Scan-Cone-Sites}:
        Keep the smallest $t^*=:a$ and smallest $t^\dagger=:b$ discs, breaking ties by diameter.
        Let, without loss of generality, $a=1$.
        Hence $b\leq 1+2/\eps_S$.
        Partition $[a,b]$ in intervals of length $\eps_C/2$, there are at most $O(1/\eps_C\eps_S )$ many of them.
        Scan discs, one at a time, and compute their diameter $t_{1,j}^*+t_{1,j}^\dagger$.
        Every disc is mapped into the respective partition of $t_{1,j}^*$.
        Every interval of this partition keeps record of a minimum diameter disc seen.

        For correctness, it suffices to show that discs that are member of the same length $\eps_C/2$ interval that are dropped, in favor of another, are $\alpha$-covered for $\alpha:=1+\eps_C$.
        In the extreme case, the declined disc has a diameter equal to the kept disc.
        We have the following two cases.
        
        In case $ t_{1,j}^* -t_{1,i}^* =\eps_C/2$, we show that
        ${t^*(s_1,s_j,w_j) \leq t^*(s_1,s_i,w_i/\alpha)}$ since the inequality for $t^\dagger$ is trivial.
        Subtracting $t_{1,i}^*$ from both sides, this holds iff
        \begin{align*}
        \eps_C/2 
        \leq    
        \frac{w_i+1}{w_i+1} \frac{\lVert s_i \rVert}{w_i/\alpha+1} - \frac{\lVert s_i \rVert}{w_i+1}
        =  \underbrace{\frac{\lVert s_i \rVert}{w_i+1}}_{\geq a =1} 
            \Big(\underbrace{ \frac{w_i+1}{w_i/\alpha+1}}_{=:x} -1\Big)~.
        \end{align*}
        Now, 
        $1+\eps_C/2 \leq x$     iff 
        $ \frac{1+\eps_C/2}{\alpha}w_i + (1+\eps_C/2)\leq w_i + 1$  iff 
        $ \frac{1+\eps_C/2}{\alpha} + \frac{\eps_C/2}{w_i} \leq  1$.
        Since $1+\eps_C < 1+\eps_S \leq w_i$, 
        it suffices to observe that
        $\frac{1+2\eps_C/2}{1+\eps_C} \leq 1$.
            
        In case $ t_{1,i}^* -t_{1,j}^* =\eps_C/2$ (i.e. $t_{1,j}^\dagger - t_{1,i}^\dagger=\eps_C/2$), the inequality for $t^*$ is trivial and it suffices to show 
        $t^\dagger(s_1,s_j,w_j) \leq t^\dagger(s_1,s_i,w_i/\alpha)$.
        Subtracting $t_{1,i}^\dagger$ from both sides, this holds iff
        \begin{align*}
        \eps_C/2 
        \leq    
        \frac{w_i-1}{w_i-1} \frac{\lVert s_i \rVert}{w_i/\alpha-1} - \frac{\lVert s_i \rVert}{w_i-1}
        =  \underbrace{\frac{\lVert s_i \rVert}{w_i-1}}_{\geq a =1} 
            \Big(\underbrace{ \frac{w_i-1}{w_i/\alpha-1}}_{=:y} -1\Big)~.
        \end{align*}
        Now, $1+\eps_C/2 \leq y$ iff 
        $ \frac{1+\eps_C/2}{\alpha}w_i - (1+\eps_C/2)\leq w_i - 1$  iff 
        $ \frac{1+\eps_C/2}{\alpha} - \frac{\eps_C/2}{w_i} \leq  1$.
        Since $w_i>0$, it suffices to observe that $1+\eps_C/2\leq \alpha$ is a true statement.
        \end{proof}

\fi
\end{document}